\documentclass[]{article}

 \usepackage{amssymb} 
\usepackage{mathtools}
\DeclarePairedDelimiter\ceil{\lceil}{\rceil}
\DeclarePairedDelimiter\floor{\lfloor}{\rfloor}
\usepackage[textsize=footnotesize]{todonotes}
\usepackage{epstopdf}% To incorporate .eps illustrations using PDFLaTeX, etc.
\usepackage[caption=false]{subfig}% Support for small, `sub' figures and tables
\usepackage[numbers,sort&compress]{natbib}% Citation support using natbib.sty
\bibpunct[, ]{[}{]}{,}{n}{,}{,}% Citation support using natbib.sty
\usepackage{capt-of} % required for \captionof
\usepackage{graphicx}
\usepackage{color}
\usepackage{anyfontsize} % remove font size warnings
\usepackage{microtype}      % microtypography

\usepackage{changepage} % required to indent quotation with adjustwidth

\usepackage{amsthm}
\theoremstyle{plain}% Theorem-like structures provided by amsthm.sty
\newtheorem{theorem}{Theorem}[section]
\newtheorem{lemma}[theorem]{Lemma}

\newtheorem{conjecture}[theorem]{Conjecture}

\theoremstyle{definition}
\newtheorem{definition}[theorem]{Definition}

\theoremstyle{remark}

\makeatletter
\def\mathcolor#1#{\@mathcolor{#1}}
\def\@mathcolor#1#2#3{%
  \protect\leavevmode
  \begingroup
    \color#1{#2}#3%
  \endgroup
}
\makeatother

\newcommand\gt[5]{
    ${\scriptstyle
    \mathcolor[rgb]{0.8000, 0.0000, 0.0000}{#1}
    \mathcolor[rgb]{0.9882, 0.6863, 0.2431}{#2}
    \mathcolor[rgb]{0.4471, 0.6235, 0.8118}{#3}
    \mathcolor[rgb]{0.4471, 0.6235, 0.8118}{\textbf{#4}}
    \mathcolor[rgb]{0.9882, 0.6863, 0.2431}{\textbf{#5}}
    }$
}

\newcommand\gf[5]{
    ${\scriptstyle
    f(
    \mathcolor[rgb]{0.8000, 0.0000, 0.0000}{#1}
    \mathcolor[rgb]{0.9882, 0.6863, 0.2431}{#2}
    \mathcolor[rgb]{0.4471, 0.6235, 0.8118}{#3}
    \mathcolor[rgb]{0.4471, 0.6235, 0.8118}{\textbf{#4}}
    \mathcolor[rgb]{0.9882, 0.6863, 0.2431}{\textbf{#5}}
    )
    }$
}

\newcommand\g[1]{
    {\footnotesize \emph{#1}}
}

\newcommand\vfigbegin[0]{
    \begin{center}
    \begin{minipage}{\textwidth}
    \begin{center}
}

\newcommand\vfigend[3]{
    \captionof{figure}[#3]{#2}
    \label{#1}
    \end{center}
    \end{minipage}
    \end{center}
}

\newcommand\image[3]{% width, height, filename
    \includegraphics[width=#1,height=#2]{#3}
}

\begin{document}

\title{Quasiperiodic bobbin lace patterns}

\author{
Veronika Irvine\\
\small Cheriton School of Computer Science\\[-0.8ex]
\small University of Waterloo\\[-0.8ex]
\small Waterloo ON, Canada\\
\and
Therese Biedl\\
\small Cheriton School of Computer Science\\[-0.8ex]
\small University of Waterloo\\[-0.8ex]
\small Waterloo ON, Canada\\
\and
Craig S. Kaplan\\
\small Cheriton School of Computer Science\\[-0.8ex]
\small University of Waterloo\\[-0.8ex]
\small Waterloo ON, Canada\\
}

\maketitle

\begin{abstract}
Bobbin lace is a fibre art form in which threads are braided together to form a fabric, often with a very detailed and complex design.  In traditional practice, each region of the fabric is filled with a periodic texture.  We establish the groundwork for non-periodic lace patterns and present three new quasiperiodic families based on Sturmian words, the Penrose tiling by thick and thin rhombs and the Ammann-bar decoration of the Penrose tiling.
\end{abstract}

\section{Introduction}

\vfigbegin
    \image{0.7\columnwidth}{!}{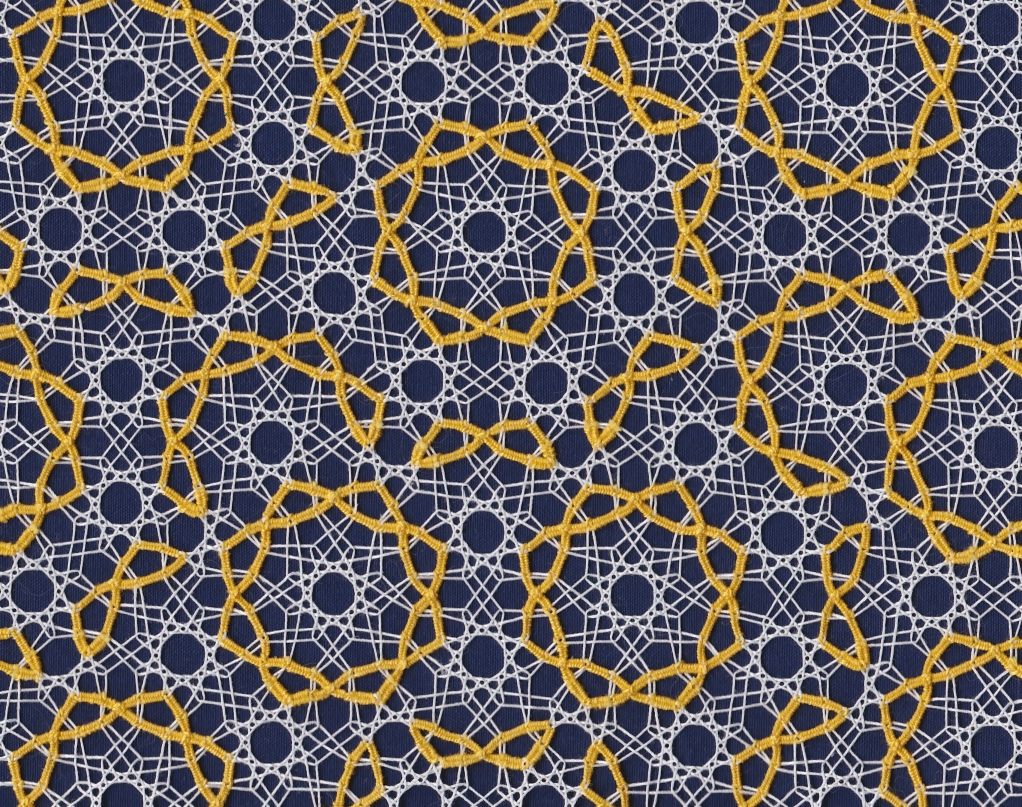}
\vfigend{fig:lace2}{`Nodding bur-marigold', Veronika Irvine 2019: dual of Ammann-Beenker tiling worked in DMC~Cordonnet~Special~80 cotton thread.}{Lace Ammann-Beenker}

Bobbin lace is a 500-year-old fibre art in which threads are braided together to form a fabric characterized by many holes.
In bobbin lace, a closed region, which may be part of the background or, less commonly, part of the foreground\footnote{Traditionally, regions of a foreground figure or motif are filled with a plain (cloth-stitch) or triaxial (half-stitch) weave or a combination of both. Contemporary artists, such as Pierre Fouch\'{e}, are starting to include decorative patterns in foreground regions.}, is filled with a decorative pattern.  A variety of patterns are used to distinguish different forms, provide shading and create aesthetic interest.
Traditionally, these patterns, called \emph{grounds} or \emph{fillings} by lacemakers, consist of a small design repeated at regular intervals horizontally and vertically.
Contemporary lace pieces also employ very regular grounds, although modern lacemakers have experimented with `random' fillings in which stitches in a traditional ground are performed without a fixed sequence~\cite[Chapter~17]{moderne}.
For a more detailed description of bobbin lace and an example pattern, we refer the reader to previous work by Irvine and Ruskey~\cite{lace1}.

Regular fillings with crystallographic symmetry create pleasing designs and in small areas provide sufficient complexity to be interesting.  However, modern lacemakers are tending towards large scale, abstract pieces with sizeable areas of a homogeneous texture.  At this scale, the simplicity of translational repetition can dominate the piece, overwhelming the more subtle reflection and rotation symmetries.  From a distance, the viewer perceives only a regular grid of holes which, on closer inspection, may give rise to a perception of motion due to a simultaneous lightness contrast illusion similar to the scintillating grid~\cite{ninio} (in particular when realized in white thread on a black background).
We are interested in breaking up the integral regularity of translation in the pattern while still maintaining local areas of symmetry. This characteristic appears in the famous Penrose tilings, which belong to the larger family of quasiperiodic tilings.  We will formally define quasiperiodicity in Section~\ref{sec:problem} but for now we can think of it informally as uniform repetition within a pattern that globally lacks periodic symmetry.

A quasiperiodic tiling contains small patches of recognizable symmetry, providing a familiar handle for the viewer to approach the design. The recognizable patches are arranged in a large variety of ways challenging the viewer to identify them.
Several philosophers have attributed aesthetic appeal to just such a relationship.  The 18th century Scottish philosopher Hutcheson proposed a calculus of beauty which he described as ``\emph{the compound ratio of uniformity and variety}'' \cite{hutcheson}.
Another 18th century British philosopher, Gerard, ascribed aesthetic pleasure to an activity of the imagination:
\begin{adjustwidth}{1cm}{1cm}
\emph{``\dots uniformity, when perfect and unmixed, is apt to pall upon the sense, to grow languid, and to sink the mind into an uneasy state of indolence.  It cannot therefore alone produce pleasure, either very high, or of very long duration. \dots Variety in some measure gratifies the sense of novelty, as our ideas vary in passing from the contemplation of one part to that of another. This transition puts the mind in action, and gives it employment, the consciousness of which is agreeable.''}~\cite{gerard}
\end{adjustwidth}
Closer to our own time, art historian Gombrich explored the importance of decoration and conjectured:``\emph{Aesthetic delight lies somewhere between boredom and confusion}.''~\cite{gombrich}

In this paper we explore the role that quasiperiodicity can play in lacemaking.  To do that, we extend the mathematical model for bobbin lace so that it can describe non-periodic patterns.  We look at the P3 tiling, consisting of thick and thin rhombs, as well as the Ammann bar decorations of this tiling and derive workable bobbin lace patterns from them.  Several new quasiperiodic lace patterns are then realized in thread.

\section{Problem description}
\label{sec:problem}
Let us start by briefly summarizing previous work on creating a mathematical model for bobbin lace grounds \cite{irvinePhD, lace1, lace2}.
A bobbin lace ground is a small pattern used to fill a closed shape.
Although the lace that includes this ground is obviously finite in size, rarely larger than a tablecloth, we can imagine that the ground may extend arbitrarily far in any direction and is therefore unbounded. Traditionally, the ground fills this space by translating a small rectangular patch of the pattern in two orthogonal directions such that the copies fit together edge to edge. In the mathematics of tiling theory, such patterns are called \emph{periodic}.

Bobbin lace is created by braiding four threads at a time.  Threads travel from one 4-stranded braid to the next in pairs.  We can therefore divide the ground into two independent components: a drawing that captures the flow of pairs of threads from one braid to another, and a description of the braid formed each time four threads, or more specifically two pairs, meet (as shown in Figure~\ref{fig:graph}.

More formally, we can represent a bobbin lace ground as the pair $(\Gamma(G),\zeta(v))$ where $G$ is a directed graph embedding that captures the flow of pairs of threads from one braid to another,
$\Gamma(G)$ is  a  specific  drawing  of $G$ that  assigns  a  position  to every vertex,  and $\zeta(v)$  is  a  mapping  from  a  vertex $v \in V(G)$  to a mathematical braid word which specifies the over and under crossings performed on the subset of four threads meeting at $v$.

\vfigbegin
    % width, filename
    \def\svgwidth{0.9\columnwidth}
%% Accompanies image file 'graph2.pdf' (pdf, eps, ps)
\begingroup%
  \makeatletter%
  \providecommand\color[2][]{%
    \errmessage{(Inkscape) Color is used for the text in Inkscape, but the package 'color.sty' is not loaded}%
    \renewcommand\color[2][]{}%
  }%
  \providecommand\transparent[1]{%
    \errmessage{(Inkscape) Transparency is used (non-zero) for the text in Inkscape, but the package 'transparent.sty' is not loaded}%
    \renewcommand\transparent[1]{}%
  }%
  \providecommand\rotatebox[2]{#2}%
  \newcommand*\fsize{\dimexpr\f@size pt\relax}%
  \newcommand*\lineheight[1]{\fontsize{\fsize}{#1\fsize}\selectfont}%
  \ifx\svgwidth\undefined%
    \setlength{\unitlength}{377.26940918bp}%
    \ifx\svgscale\undefined%
      \relax%
    \else%
      \setlength{\unitlength}{\unitlength * \real{\svgscale}}%
    \fi%
  \else%
    \setlength{\unitlength}{\svgwidth}%
  \fi%
  \global\let\svgwidth\undefined%
  \global\let\svgscale\undefined%
  \makeatother%
  \begin{picture}(1,0.2824646)%
    \lineheight{1}%
    \setlength\tabcolsep{0pt}%
    \put(0,0){\includegraphics[width=\unitlength,page=1]{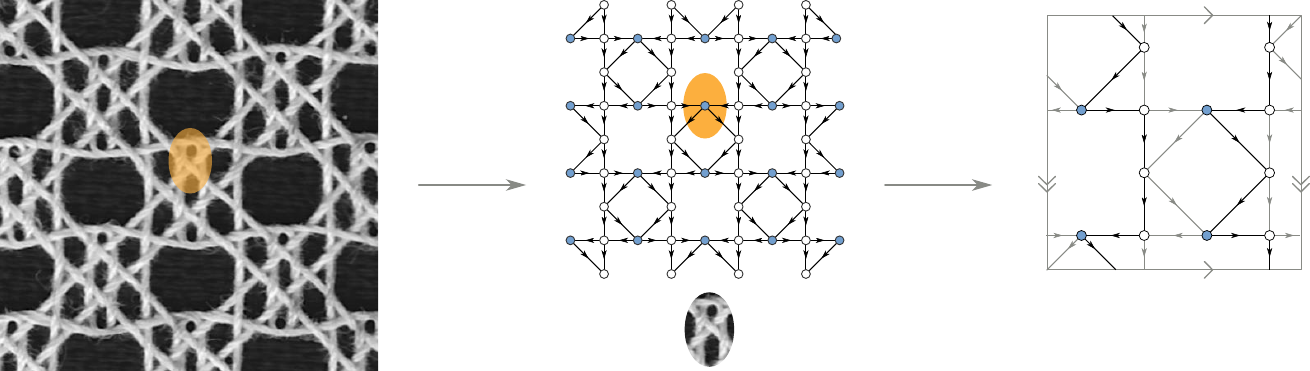}}%
    \put(0.68755481,0.19133896){\color[rgb]{0,0,0}\makebox(0,0)[lt]{\lineheight{1.25}\smash{\begin{tabular}[t]{l}\small{$\Gamma(G) =$}\end{tabular}}}}%
    \put(0.79886301,0.0440272){\color[rgb]{0,0,0}\makebox(0,0)[lt]{\lineheight{1.25}\smash{\begin{tabular}[t]{l}\small{$\zeta(v_{blue}) = CTCT$}\end{tabular}}}}%
    \put(0.79886301,0.00824376){\color[rgb]{0,0,0}\makebox(0,0)[lt]{\lineheight{1.25}\smash{\begin{tabular}[t]{l}\small{$\zeta(v_{white}) = CT$}\end{tabular}}}}%
  \end{picture}%
\endgroup%

\vfigend{fig:graph}{From bobbin lace ground to mathematical representation as a graph drawing and a set of braid words, one for each vertex in the graph}{Graph drawing representation}

We are interested in ground patterns that can actually be realized by a lacemaker.  To this end, Irvine and Ruskey identified four necessary conditions that the drawing $\Gamma(G)$ must meet:

\begin{itemize}
\item[CP1.] \textbf{2-2-regular digraph:} Two pairs come together to form a braid and, when the braid is complete, the four threads continue on, pair-wise, to participate in other braids.  Therefore, $G$ is a directed 2-2-regular digraph, meaning that every vertex has two incoming edges and two outgoing edges.
\item[CP2.] \textbf{Connected filling of unbounded size:} The purpose of a bobbin lace ground is to fill a simple region of any size with a continuous fabric.  One way to accommodate large sizes is to choose a pattern that can cover the infinite plane.  In traditional lace, this is accomplished with a periodic pattern.  As shown on the right side of Figure~\ref{fig:graph}, a parallelogram with the smallest area that captures the translational period of the pattern, called a unit cell, can represent the entire pattern. Threads leaving the bottom edge of the parallelogram continue as threads entering the top edge of the next repeat; threads leaving the right edge enter the next repeat on the left edge and vice versa.
    We refer to this representation as a \emph{flat torus}.
    For the resulting fabric to hold together as one piece, it is important that translated copies of the pattern are connected to each other.  We express this formally by saying that the graph embedding of $G$ must have an oriented genus of 1 (i.e. it must be an embedding on the torus).
\item[CP3.] \textbf{Partially ordered:} Bobbin lace is essentially an alternating braid and therefore the thread crossings must have a partial order, or put another way, all thread crossings must happen in the forward direction.  For the graph embedded on the torus, this means all directed circuits of $G$ are non-contractible.  A contractible circuit can be reduced to a point by shortening the lengths of its edges.
\item[CP4.] \textbf{Thread conserving:} Loose ends, caused by cutting threads or adding new ones, are undesirable because they inhibit the speed of working, can fray or stick out in an unsightly manner and, most importantly, degrade the strength of the fabric. For the model, this means that, once started, a rectangular region of arbitrary width $w$ can be worked to any length without the addition or termination of threads (assuming that sufficient thread has been wound around the bobbins). We refer to this ability to extend the pattern indefinitely in one direction, using a fixed set of threads, as ``conservation of threads''.
    We say that a pair of paths in a graph drawing is \emph{osculating} if, when the two paths meet, they do not cross transversely but merely kiss, i.e., touch and continue, without crossing, as shown in Figure~\ref{fig:partition}.
    In order to ensure that the pattern conserves threads, when the edges of $\Gamma(G)$ are partitioned into directed osculating circuits we must ensure that each circuit in the partition is in the $(1,0)$-homotopy class of the torus (Adams gives an excellent  description of the torus-knot naming convention~\cite{adams}).  That is, each circuit must wrap once around the minor (vertical) axis of the flat torus and zero times around its major (non-vertical) axis.  This ensures that the threads do not have a net drift left or right but rather return to the same horizontal position at the start of each vertical repeat.
\end{itemize}

 \begin{center}
 \begin{minipage}{\textwidth}
 \begin{center}
 % width, filename
    \def\svgwidth{0.9\columnwidth}
%% Accompanies image file 'conservation_periodic.pdf' (pdf, eps, ps)
\begingroup%
  \makeatletter%
  \providecommand\color[2][]{%
    \errmessage{(Inkscape) Color is used for the text in Inkscape, but the package 'color.sty' is not loaded}%
    \renewcommand\color[2][]{}%
  }%
  \providecommand\transparent[1]{%
    \errmessage{(Inkscape) Transparency is used (non-zero) for the text in Inkscape, but the package 'transparent.sty' is not loaded}%
    \renewcommand\transparent[1]{}%
  }%
  \providecommand\rotatebox[2]{#2}%
  \newcommand*\fsize{\dimexpr\f@size pt\relax}%
  \newcommand*\lineheight[1]{\fontsize{\fsize}{#1\fsize}\selectfont}%
  \ifx\svgwidth\undefined%
    \setlength{\unitlength}{216bp}%
    \ifx\svgscale\undefined%
      \relax%
    \else%
      \setlength{\unitlength}{\unitlength * \real{\svgscale}}%
    \fi%
  \else%
    \setlength{\unitlength}{\svgwidth}%
  \fi%
  \global\let\svgwidth\undefined%
  \global\let\svgscale\undefined%
  \makeatother%
  \begin{picture}(1,0.4111029)%
    \lineheight{1}%
    \setlength\tabcolsep{0pt}%
    \put(0,0){\includegraphics[width=\unitlength,page=1]{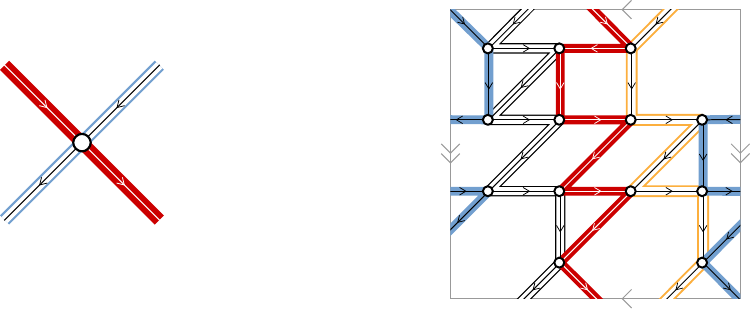}}%
    \put(0.10918848,0.05429235){\color[rgb]{0,0,0}\makebox(0,0)[t]{\lineheight{1.25}\smash{\begin{tabular}[t]{c}\g{transverse}\end{tabular}}}}%
    \put(0,0){\includegraphics[width=\unitlength,page=2]{conservation_periodic.pdf}}%
    \put(0.40310972,0.05409562){\color[rgb]{0,0,0}\makebox(0,0)[t]{\lineheight{1.25}\smash{\begin{tabular}[t]{c}\g{osculating}\end{tabular}}}}%
  \end{picture}%
\endgroup%

 \captionof{figure}[]{At a 2-in, 2-out vertex, two paths meet in either a transverse or an osculating manner.  An example of a graph drawing partitioned into osculating paths.}
 \label{fig:partition}
 \end{center}
 \end{minipage}
 \end{center}

%\vfigbegin
%    \imagepdftex{0.9\columnwidth}{conservation_periodic.pdf_tex}
%\vfigend{fig:partition}{At a 2-in, 2-out vertex, two paths meet in either a transverse or an osculating manner.  An example of a graph drawing partitioned into osculating paths.}

For a more detailed description of these conditions, we refer the reader to previous work by Irvine and Ruskey~\cite{lace1} or Biedl and Irvine~\cite{lace2}.
Irvine used conditions $CP1$--$CP4$ to prove that there exist an infinite number of periodic bobbin lace grounds~\cite{irvinePhD}.  She also performed a combinatorial search to generate several million examples for graph drawings with up to 20 vertices~\cite{irvinePhD}.
However, the conditions introduced previously are only applicable to periodic patterns. Here, we introduce a new set of conditions that cover any infinite pattern, including those that are non-periodic.

Although we represent the flow of threads as a drawing of one-dimensional curves, lace is realized in a physical medium (such as cotton thread) which has width.  There are physical limitations on how close together two braids can be before they will overlap or how close two pins can be before the braid between them buckles out of plane.  Conversely, if a hole in the lace ground is too wide, it may be wider than the region we are trying to fill, slicing the patch of fabric into two unconnected parts.
We must therefore specify a set of minimum and maximum feature sizes, namely a lower and upper bound on the length of an edge and a lower and upper bound on the distance between two vertices of a face.
 In previous work on periodic grounds, we took these bounds for granted because they are determined by the boundary of the unit cell and $CP4$.  Here we will include them in a formal way by first defining a few terms:

\begin{definition}\cite{grunbaum}
The faces of a graph drawing are \emph{uniformly bounded} if there exist positive real numbers $r$ and $R$, $r < R$, such that every face contains a ball of radius $r$ and every face is contained in a ball of radius $R$.
\end{definition}

In any tiling with a finite number of distinct tile shapes (or specifically, a graph drawing with a finite number of distinct face shapes), the tiles will be uniformly bounded.

\begin{definition}\cite{senechal}
A \emph{Delone point set} is uniformly discrete and relatively dense.
A set of points $\mathcal{P}$ in $\mathbb{R}^n$ is \emph{uniformly discrete} if there exists a positive real number $d$ such that given any two points $x$ and $y$ in $\mathcal{P}$, the distance between them is at least $2d$.
A set of points $\mathcal{P}$ in $\mathbb{R}^n$ is \emph{relatively dense} if there exists a positive real number $D$ such that every ball with radius greater than $D$ contains at least one point of $\mathcal{P}$ in its interior.
\end{definition}

\begin{itemize}
\item[C0.] \textbf{Bounded in feature size:} The faces of the graph drawing are uniformly bounded and the vertices form a Delone
point set.
\end{itemize}

The remaining conditions are a generalization of $CP1$ --- $CP4$.

\begin{itemize}
\item[C1.] \textbf{2-2-regular digraph:} The underlying graph $G$, representing the flow of pairs of threads between braids, is a 2-2-regular digraph.
\item[C2.] \textbf{Connected filling of unbounded size:} The lace must fill a 2D region of unbounded size with a continuous fabric.  The graph $G$ must therefore be infinite (also implied by $C0$) and connected.
\item[C3.] \textbf{Partially ordered:} Bobbin lace is braided, which means that all crossings must happen in the forward direction; therefore, the planar embedding of $G$ cannot contain a directed cycle.  This implies that the planar embedding of $G$ is simple (faces have degree at least 3): it does not contain any self-loops; either the edges of a bigon (a face of degree two) form a cycle or the edges and vertices of a bigon represent a continuous braid made on the same four threads which, by our mapping, corresponds to a single vertex in the graph.
\item[C4.] \textbf{Thread conserving:} For threads to be conserved, there must exist a partition of the plane graph drawing, $\Gamma(G)$, into a set of \emph{well-behaved} osculating paths.  A path is well-behaved if there exists a line $\ell$ and a finite distance $s$ such that every point on the path is within a perpendicular distance $s$ of $\ell$.
\end{itemize}

It is worth noting that a necessary condition for $C3$ is that the outgoing edges of an infinite digraph embedded in the plane must be \emph{rotationally consecutive} (i.e., adjacent to one another in clockwise order around their common endpoint), a result that was demonstrated by Irvine and Ruskey~\cite{lace1}.  This rotationally consecutive ordering coupled with condition $C1$--a regular digraph with in-degree equal to out-degree--also ensures that we can partition the graph into a unique set of osculating paths~\cite{lace1}.

Conditions $C0$--$C4$ give a generalized model of bobbin lace grounds that encompass both periodic and non-periodic patterns. Every lace pattern can be represented as a planar infinite graph.  Every simple planar graph can be rendered as a straight-line graph drawing~\cite{fary}, which in turn defines a tiling of the plane.  We can now define a quasiperiodic tiling and ask the question `Does there exist a quasiperiodic drawing of an infinite graph that meets conditions $C0$--$C4$?'.

\begin{definition}\cite{durand, delvenne}
A \emph{quasiperiodic tiling} is a tiling $T$ such that for every patch $P$ (a finite, simply connected subset of tiles in $T$), there exists a real number $b>0$ such that a ball of radius $b$, centered on any point in the tiling,
contains a copy of $P$.  In addition, the tiling is not periodic, that is, a copy $T'$ of $T$ cannot be superimposed on $T$ by a non-trivial translation.
\end{definition}

In the rest of this paper, we will look at different families of quasiperiodic patterns.  When trying to assess whether a quasiperiodic pattern is a good candidate for a bobbin lace ground, we will pay particular attention to two conditions: Do all of the vertices have degree four ($C1$)?  Can we assign a direction to the edges to give a well-behaved osculating partition ($C4$)?

In the next section, we start our exploration with a very simple family of non-periodic patterns composed from parallelograms arranged in a grid.

\section{Parallelogram tiling from two sets of parallel lines}
\label{sec:two}
Obviously a periodic grid will satisfy $C0$--$C4$ so we ask the question ``What minimal alterations will break periodicity while still meeting these conditions?''

To start, consider a simple periodic bigrid formed by overlaying two infinite sets of regularly spaced parallel lines, $A$ and $B$, at a relative angle of $\alpha > 0$, where $\alpha$ is the small angle between lines $A_i$ and $B_j$ for any $A_i \in A$ and $B_j \in B$.
To turn this into a planar graph drawing, we place a vertex at every crossing of two lines and assign an edge between pairs of vertices that are consecutive along a line.  We observe that the faces in this drawing are all parallelograms and the planarity of the drawing satisfies $C2$.

To assign a direction to each edge of the drawing, we first choose a vector that is not perpendicular to any set of lines and then rotate the drawing so that this vector points up.  We can now unambiguously assign a downward direction to each edge because there are no horizontal edges.  The directed graph drawing thus created has two useful properties.  First, each vertex has two incoming and two outgoing edges such that the outgoing edges are rotationally consecutive, satisfying $C1$.  Second, no directed cycle can exist because there are no edges directed `up' to complete the cycle, satisfying condition $C3$.\footnote{The edge directions specified here fix a start and end vertex for each edge.  Once the edge directions have been assigned, rotating the drawing such that an edge is no longer strictly pointing down is not a problem because the topological direction of the edges is still preserved.}

Adherence to conditions $C1$--$C3$ is determined by the topology of the pattern.  As long as any change we make to the geometry of the drawing does not alter its topology, the drawing will still meet these conditions.
We now consider modifications to the geometry.  We can turn our drawing into a non-periodic pattern by changing the spacing between lines.
For simplicity, we will restrict the allowed spacings to two values: long $(L)$ and short $(S)$.  Under this restriction, the faces in the drawing are limited to four parallelogram classes: $SS$, $LL$, $SL$ and $LS$.  The largest circle that fits inside every tile of the tile set has radius $(|S|\sin{\alpha})/2$ and the smallest circle that can contain every tile of the tile set has radius less than $|L|$.  No pair of vertices in the tiling are closer to each other than $|S|/\sin{\alpha}$, and every vertex is at most $2|L|$ from some other vertex.  Since $\alpha$ is greater than $0$, the vertices form a Delone point set, the faces are uniformly bounded, and $C0$ is satisfied.

Our goal is to create a pattern that is not periodic but that still has some regularity.  To achieve this, instead of randomly choosing whether to use $L$ or $S$ spacing between two lines in a set, we will use a quasiperiodic word.
An infinite word $W$ is \emph{quasiperiodic} if and only if, for every finite substring $w$ of $W$, there exists an integer $b$ such that every substring of length $b$ in $W$ contains at least one copy of $w$~\cite{delvenne}.  We can think of it as a quasiperiodic tiling in one dimension.

In Figure~\ref{fig:bigrid}, we show a pattern that uses the infinite binary Fibonacci word to determine the spacing.
The sequence of characters in the Fibonacci word is defined by the following recursive expansion relation: $\sigma(L)=LS$, $\sigma(S)=L$. For example, starting with the word $L$, we have $\sigma(L)=LS$, $\sigma^2(L)=LSL$, $\sigma^3(L)=LSLLS$, $\sigma^4(L)=LSLLSLSL$ and so on.

 \begin{center}
 \begin{minipage}{\textwidth}
 \begin{center}
 % width, filename
    \def\svgwidth{0.6\columnwidth}
%% Accompanies image file 'bigrid2.pdf' (pdf, eps, ps)
\begingroup%
  \makeatletter%
  \providecommand\color[2][]{%
    \errmessage{(Inkscape) Color is used for the text in Inkscape, but the package 'color.sty' is not loaded}%
    \renewcommand\color[2][]{}%
  }%
  \providecommand\transparent[1]{%
    \errmessage{(Inkscape) Transparency is used (non-zero) for the text in Inkscape, but the package 'transparent.sty' is not loaded}%
    \renewcommand\transparent[1]{}%
  }%
  \providecommand\rotatebox[2]{#2}%
  \newcommand*\fsize{\dimexpr\f@size pt\relax}%
  \newcommand*\lineheight[1]{\fontsize{\fsize}{#1\fsize}\selectfont}%
  \ifx\svgwidth\undefined%
    \setlength{\unitlength}{415.6162262bp}%
    \ifx\svgscale\undefined%
      \relax%
    \else%
      \setlength{\unitlength}{\unitlength * \real{\svgscale}}%
    \fi%
  \else%
    \setlength{\unitlength}{\svgwidth}%
  \fi%
  \global\let\svgwidth\undefined%
  \global\let\svgscale\undefined%
  \makeatother%
  \begin{picture}(1,0.74869458)%
    \lineheight{1}%
    \setlength\tabcolsep{0pt}%
    \put(0,0){\includegraphics[width=\unitlength,page=1]{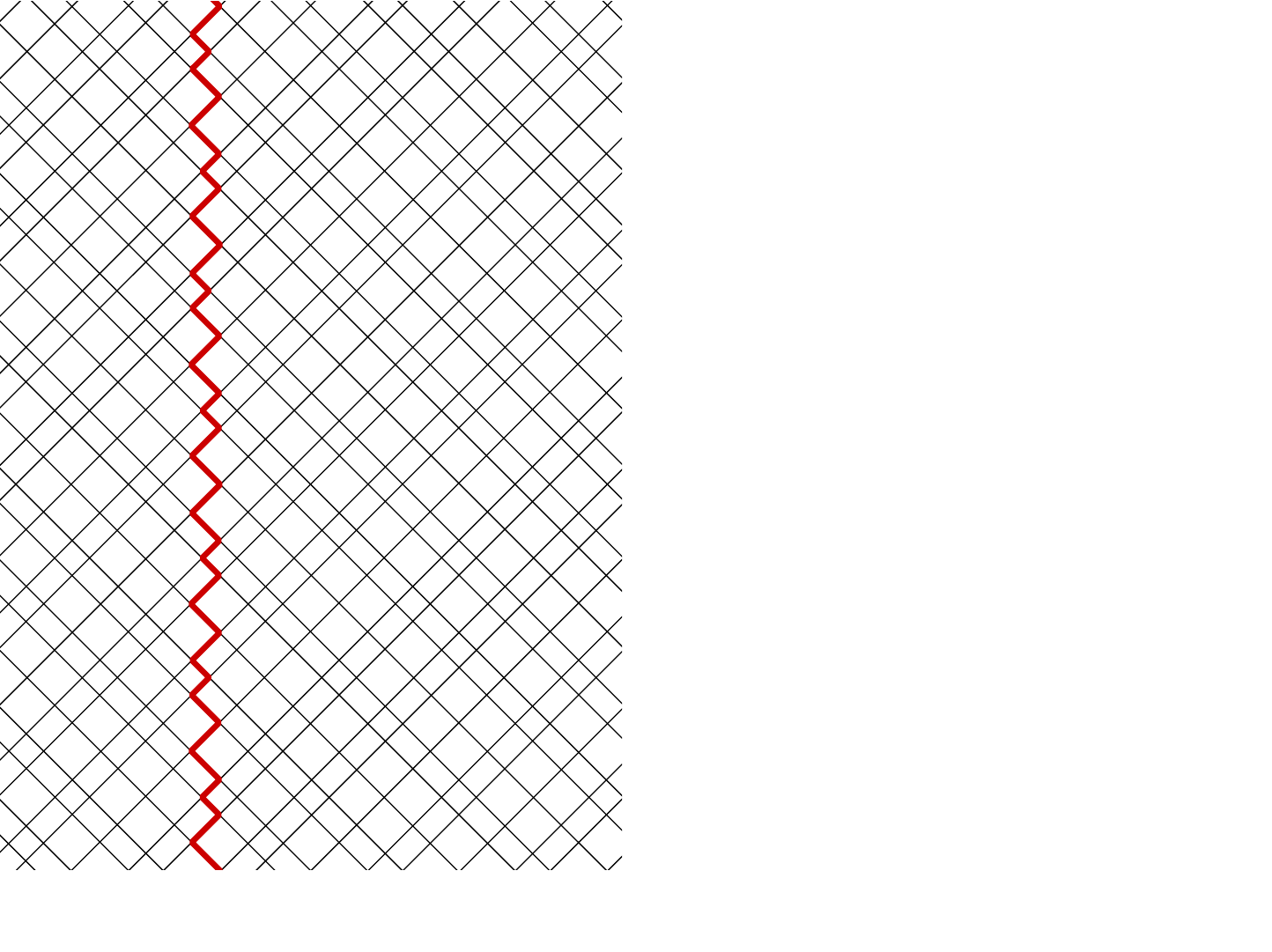}}%
    \put(0.24502624,0.00607627){\color[rgb]{0,0,0}\makebox(0,0)[t]{\lineheight{1.25}\smash{\begin{tabular}[t]{c}\g{(a)}\end{tabular}}}}%
    \put(0.75481993,0.00607627){\color[rgb]{0,0,0}\makebox(0,0)[t]{\lineheight{1.25}\smash{\begin{tabular}[t]{c}\g{(b)}\end{tabular}}}}%
    \put(0,0){\includegraphics[width=\unitlength,page=2]{bigrid2.pdf}}%
  \end{picture}%
\endgroup%

 \captionof{figure}[]{A simple quasiperiodic lace pattern created by superimposing two sets of parallel lines with interline spacing determined by the Fibonacci word: a) pattern with one path from osculating partition highlighted in bold red, b) pattern worked as bobbin lace, tallies (fat braids) at vertices incident to four squares of length $L$.}
 \label{fig:bigrid}
 \end{center}
 \end{minipage}
 \end{center}

%\vfigbegin
%    \imagepdftex{0.6\columnwidth}{bigrid2.pdf_tex}
%\vfigend{fig:bigrid}{A simple quasiperiodic lace pattern created by superimposing two sets of parallel lines with interline spacing determined by the Fibonacci word: a) pattern with one path from osculating partition highlighted in bold red, b) pattern worked as bobbin lace, tallies (fat braids) at vertices incident to four squares of length $L$.}

To test compliance with $C4$, we will use the  \emph{Hamming weight} of a binary word, which can be defined here as the number of occurrences of $S$ within that word.  For example, substring $w=LLSLLSLSL$ of the Fibonacci word has Hamming weight $3$.  Lothaire showed that the Fibonacci word is \emph{balanced}: given any two equal length substrings, their Hamming weights differ by at most one~\cite{lothaire}.

We partition our graph drawing into a set of osculating paths and consider one such path $P$.  If $P$ follows two consecutive edges in line set $A$, it must transversely cross another path, contradicting the definition of an osculating path, therefore, the steps in $P$ alternate between an edge from line set $A$ and an edge from line set $B$.

Consider $\ell$, the least squares regression line (also known as the Deming regression line) derived from the vertex positions of $P$.  A least squares regression line minimizes the sum, over all vertices $v$ in $P$, of the perpendicular distance from $v$ to $\ell$~\cite{glaister}.
A walk of $2k$ steps along path $P$ can be described by two equal length substrings of the Fibonacci word which we shall call $w_A$ (steps along consecutive lines in $A$) and $w_B$ (steps along consecutive lines in $B$).  Because the line spacing of our bigrid is based on a balanced word, $w_A$  and $w_B$ have the same number of $S$ steps, plus or minus one, for all values of $k$.
Therefore, for any walk along $P$, the distance travelled along lines in set $A$, in a left to right direction relative to $\ell$ is equal or very nearly equal to the distance travelled along lines in set $B$, right to left relative to $\ell$.  When a one step discrepancy occurs, it does not accumulate because $w_A$ and $w_B$ are balanced for all values of $k$.

We observe that $\ell$ is parallel to the bisector of the angle between $A$ and $B$. For the Fibonacci word, the maximum perpendicular distance from a vertex of $P$ to the line $\ell$ is $\frac{1}{2}(L+S)\sin{\alpha}$.
The final step to satisfy $C4$ is to orient the pattern so that line $\ell$ is vertical.

The Fibonacci word belongs to a larger family called the Sturmian words. This family of infinite binary words is characterized as being non-periodic and having a balanced Hamming weight~\cite{lothaire}.  Additional lace patterns can be obtained by varying the line spacing using an alternative Sturmian word, such as the Octonacci or Thue--Morse sequences~\cite{walter}.

\begin{theorem}
The graph drawing induced by a bigrid, with line spacing determined by a Sturmian word, has a set of edge directions that satisfies conditions $C0$--$C4$ and thus forms a workable bobbin lace pattern.
\end{theorem}

In Figure~\ref{fig:bigrid}, we have our first example of a quasiperiodic bobbin lace ground.  We also note that it has some local
$D_2$ and $D_4$ dihedral symmetry.

To illustrate the significance of geometry in the conservation of threads, we present
a counterexample that is homeomorphic to both the 2D periodic lattice and the 2D Fibonacci tiling discussed above.  The topology of our counterexample obeys $C1$--$C3$ and the geometry even observes $C0$.  The interline spacing, however, does not result in well-behaved osculating paths.

In our counterexample, the spacing between south-easterly lines is specified by the bi-infinite word $W_A:=\dots S^{64}L^{16}S^{4} LL S^{4}L^{16}S^{64}\dots$ and the spacing between south-westerly lines by $W_B:=\dots L^{64}S^{16}L^{4} SS L^{4}S^{16}L^{64}\dots$.
In both sequences, the number of consecutive, equally-spaced lines (i.e. consecutive $S$s or consecutive $L$s) grows exponentially out from the center, indicated by $O$ in Figure~\ref{fig:zigzag}.  The key difference between $W_A$ and $W_B$ is that they are out of phase: when $W_A$ contains a stretch of $S$s, $W_B$ contains an equal number of $L$s.  When the osculating path travels through a region of $SL$ parallelograms, it has a net negative slope; when the path moves into a region of $LS$ parallelograms, the net slope becomes positive. The area occupied by connected parallelograms of the same type grows exponentially as we move out from $O$ causing the osculating path to travel increasingly larger distances from the ideal least squares regression line.  From a bobbin lacemaker's perspective, this means that in order to use the same set of threads from top to bottom to cover a rectangular patch with this ground, the minimum width of the patch must increase as the length of the patch increases, a contradiction to $C4$.

\vfigbegin
    % width, filename
    \def\svgwidth{\columnwidth}
%% Accompanies image file 'zigzag4.pdf' (pdf, eps, ps)
\begingroup%
  \makeatletter%
  \providecommand\color[2][]{%
    \errmessage{(Inkscape) Color is used for the text in Inkscape, but the package 'color.sty' is not loaded}%
    \renewcommand\color[2][]{}%
  }%
  \providecommand\transparent[1]{%
    \errmessage{(Inkscape) Transparency is used (non-zero) for the text in Inkscape, but the package 'transparent.sty' is not loaded}%
    \renewcommand\transparent[1]{}%
  }%
  \providecommand\rotatebox[2]{#2}%
  \newcommand*\fsize{\dimexpr\f@size pt\relax}%
  \newcommand*\lineheight[1]{\fontsize{\fsize}{#1\fsize}\selectfont}%
  \ifx\svgwidth\undefined%
    \setlength{\unitlength}{247.66422844bp}%
    \ifx\svgscale\undefined%
      \relax%
    \else%
      \setlength{\unitlength}{\unitlength * \real{\svgscale}}%
    \fi%
  \else%
    \setlength{\unitlength}{\svgwidth}%
  \fi%
  \global\let\svgwidth\undefined%
  \global\let\svgscale\undefined%
  \makeatother%
  \begin{picture}(1,0.87214842)%
    \lineheight{1}%
    \setlength\tabcolsep{0pt}%
    \put(0,0){\includegraphics[width=\unitlength,page=1]{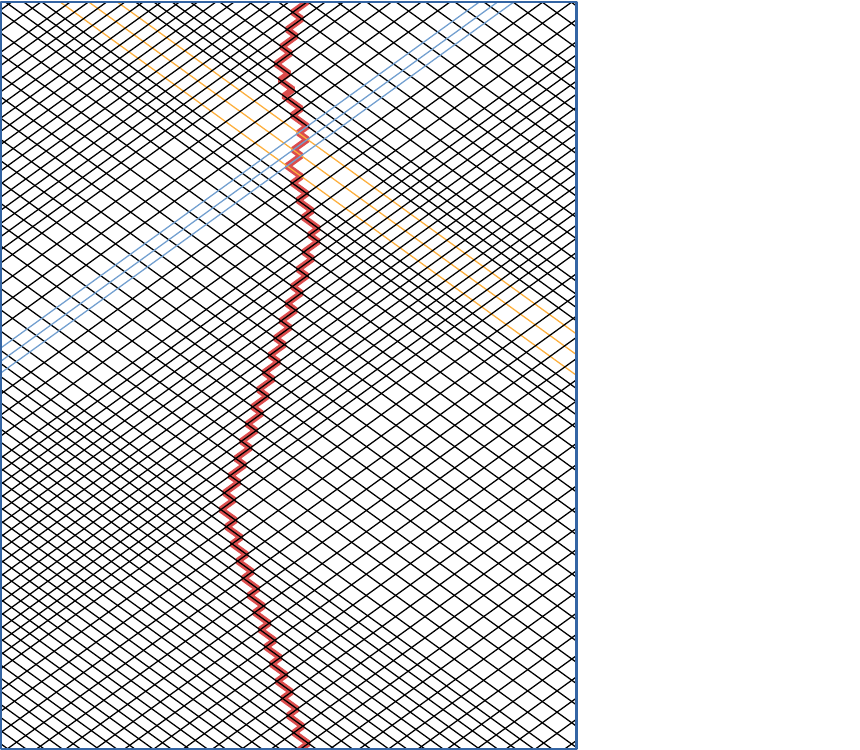}}%
    \put(0.32435373,0.69729483){\color[rgb]{0,0,0}\makebox(0,0)[t]{\lineheight{1.25}\smash{\begin{tabular}[t]{c}\g{$O$}\end{tabular}}}}%
    \put(0,0){\includegraphics[width=\unitlength,page=2]{zigzag4.pdf}}%
    \put(0.78618235,0.81042228){\color[rgb]{0,0,0}\makebox(0,0)[t]{\lineheight{1.25}\smash{\begin{tabular}[t]{c}\g{$O$}\end{tabular}}}}%
    \put(0,0){\includegraphics[width=\unitlength,page=3]{zigzag4.pdf}}%
  \end{picture}%
\endgroup%

\vfigend{fig:zigzag}{Counterexample: Two sets of parallel lines for which an osculating path deviates an increasing amount from the least squares regression line, calculated from the positions of the vertices.}{Counterexample}

As shown in Figure~\ref{fig:bigrid}, we have now achieved the first part of our goal: to demonstrate the existence of non-periodic bobbin lace grounds. Aesthetically however, the pattern of holes produced by two sets of parallel lines is quite simple.  In appearance, they resemble woven cloth; woven patterns derived from Sturmian words were previously introduced by Ahmed~\cite{ahmed}.
In bobbin lace, threads can travel in more than two directions which might permit more complex patterns.
In the next section we will demonstrate that the P3 tiling gives rise to a bobbin lace pattern.

\section{Lace pattern from the Penrose thin and thick rhombs}
The quasiperiodic tiling commonly known as the \emph{Penrose thick and thin rhomb} tiling (or, for brevity, the \emph{P3} tiling) was discovered independently by Roger Penrose and Robert Ammann in the late 1970s~\cite{grunbaum}.  One of its remarkable properties is that it has regions of fivefold symmetry which occur at larger and larger scales, ultimately leading to two particular configurations that have global 5-fold dihedral symmetry. This tendency toward fivefold symmetry also shows up in the Bragg diffraction pattern of P3.  Recall that global 5-fold symmetry cannot occur in a periodic pattern.

The set of prototiles used in P3 consists of two rhombs with interior angles as indicated in Figure~\ref{fig:penrose}a.  It is possible to create a periodic tiling from one or both of these prototiles, so the shapes are decorated in a way that prohibits periodic configurations. Such decorations are called \emph{local matching conditions}. For example, in Figure~\ref{fig:penrose}a, the edges are marked with arrows.  Only tile edges with the same arrow shape and direction can be placed together as shown in Figure~\ref{fig:penrose}b.

\vfigbegin
    % width, filename
    \def\svgwidth{\columnwidth}
%% Accompanies image file 'penrose3.pdf' (pdf, eps, ps)
\begingroup%
  \makeatletter%
  \providecommand\color[2][]{%
    \errmessage{(Inkscape) Color is used for the text in Inkscape, but the package 'color.sty' is not loaded}%
    \renewcommand\color[2][]{}%
  }%
  \providecommand\transparent[1]{%
    \errmessage{(Inkscape) Transparency is used (non-zero) for the text in Inkscape, but the package 'transparent.sty' is not loaded}%
    \renewcommand\transparent[1]{}%
  }%
  \providecommand\rotatebox[2]{#2}%
  \newcommand*\fsize{\dimexpr\f@size pt\relax}%
  \newcommand*\lineheight[1]{\fontsize{\fsize}{#1\fsize}\selectfont}%
  \ifx\svgwidth\undefined%
    \setlength{\unitlength}{424.72203827bp}%
    \ifx\svgscale\undefined%
      \relax%
    \else%
      \setlength{\unitlength}{\unitlength * \real{\svgscale}}%
    \fi%
  \else%
    \setlength{\unitlength}{\svgwidth}%
  \fi%
  \global\let\svgwidth\undefined%
  \global\let\svgscale\undefined%
  \makeatother%
  \begin{picture}(1,0.70426636)%
    \lineheight{1}%
    \setlength\tabcolsep{0pt}%
    \put(0,0){\includegraphics[width=\unitlength,page=1]{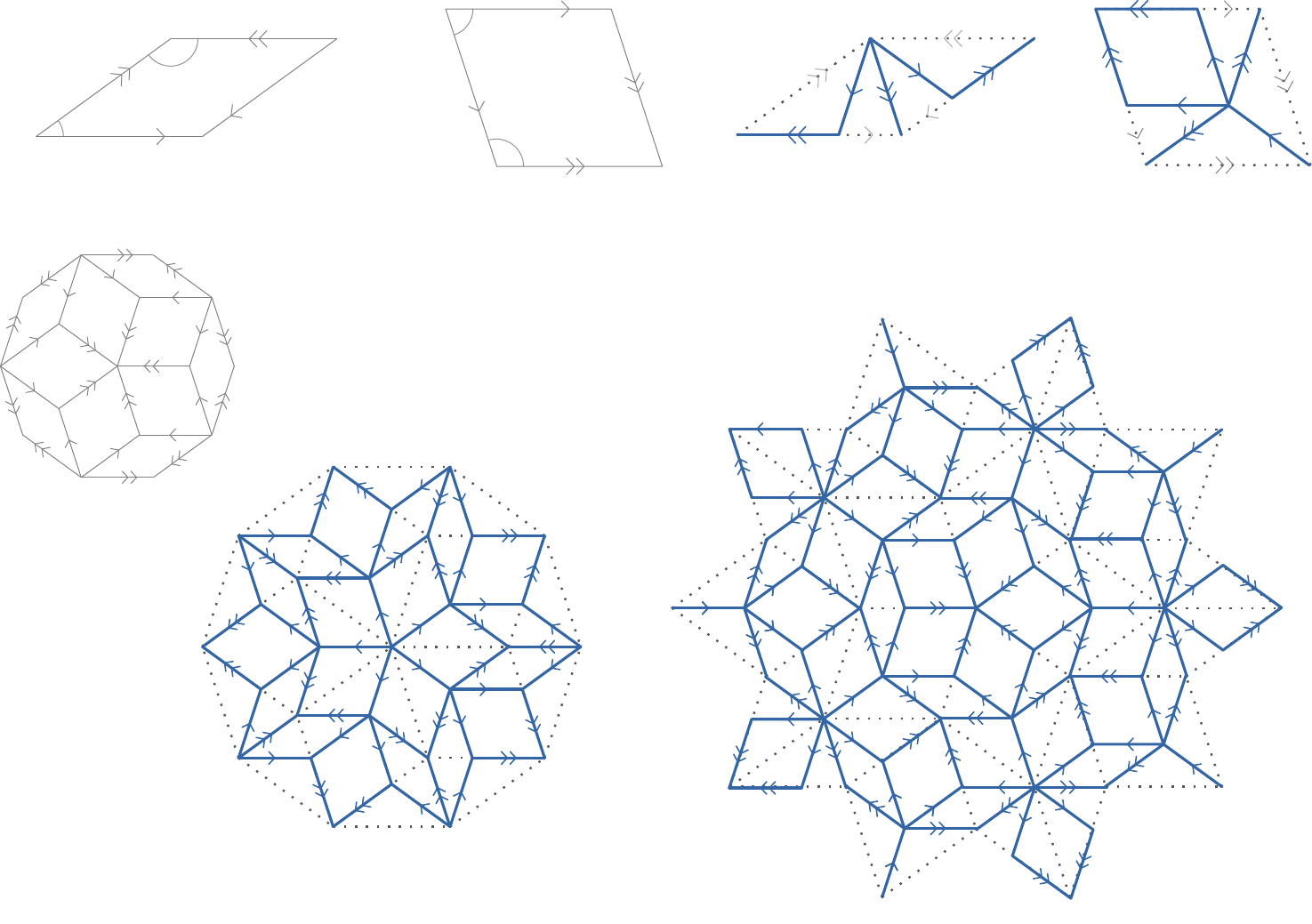}}%
    \put(0.31133927,0.68516129){\color[rgb]{0,0,0}\makebox(0,0)[t]{\lineheight{0}\smash{\begin{tabular}[t]{c}\g{$\frac{2\pi}{5}$}\end{tabular}}}}%
    \put(0.33606374,0.57856103){\color[rgb]{0,0,0}\makebox(0,0)[t]{\lineheight{0}\smash{\begin{tabular}[t]{c}\g{$\frac{3\pi}{5}$}\end{tabular}}}}%
    \put(0.09138933,0.6716826){\color[rgb]{0,0,0}\makebox(0,0)[t]{\lineheight{0}\smash{\begin{tabular}[t]{c}\g{$\frac{4\pi}{5}$}\end{tabular}}}}%
    \put(-0.00281758,0.60256876){\color[rgb]{0,0,0}\makebox(0,0)[t]{\lineheight{0}\smash{\begin{tabular}[t]{c}\g{$\frac{\pi}{5}$}\end{tabular}}}}%
    \put(0.08644225,0.57992242){\color[rgb]{0,0,0}\makebox(0,0)[t]{\lineheight{0}\smash{\begin{tabular}[t]{c}\g{$1$}\end{tabular}}}}%
    \put(0.21627473,0.61704415){\color[rgb]{0,0,0}\makebox(0,0)[t]{\lineheight{0}\smash{\begin{tabular}[t]{c}\g{$1$}\end{tabular}}}}%
    \put(0.44724607,0.55639211){\color[rgb]{0,0,0}\makebox(0,0)[t]{\lineheight{0}\smash{\begin{tabular}[t]{c}\g{$1$}\end{tabular}}}}%
    \put(0.51775857,0.63273092){\color[rgb]{0,0,0}\makebox(0,0)[t]{\lineheight{0}\smash{\begin{tabular}[t]{c}\g{$1$}\end{tabular}}}}%
    \put(0.28685616,0.54621718){\color[rgb]{0,0,0}\makebox(0,0)[t]{\lineheight{0}\smash{\begin{tabular}[t]{c}\g{(a)}\end{tabular}}}}%
    \put(0.77312822,0.54621718){\color[rgb]{0,0,0}\makebox(0,0)[t]{\lineheight{0}\smash{\begin{tabular}[t]{c}\g{(c)}\end{tabular}}}}%
    \put(0.09671108,0.28749239){\color[rgb]{0,0,0}\makebox(0,0)[t]{\lineheight{0}\smash{\begin{tabular}[t]{c}\g{(b)}\end{tabular}}}}%
    \put(0.42941791,0.004955){\color[rgb]{0,0,0}\makebox(0,0)[t]{\lineheight{0}\smash{\begin{tabular}[t]{c}\g{(d)}\end{tabular}}}}%
  \end{picture}%
\endgroup%

\vfigend{fig:penrose}{P3 tiling: a) prototiles with matching rules, b) a patch of P3, c) deflation rules d) two steps of deflation applied to the patch in (b).}{P3 decorated rhombs}

There are several ways to construct a patch of P3.  In this paper, we will look at two methods which will provide insight into how a bobbin lace pattern can be derived from this tiling.

\textbf{Deflation and matching:} A recursive process can be used to grow a small patch of tiles into a patch of unbounded size.  It starts with a finite configuration known to appear in the pattern.  For example, for P3 it could be as simple as a single thin or thick rhomb.  At each iteration, the patch is scaled up by a factor, which for P3 is $\tau = (1+\sqrt{5})/2$,  and each tile in the patch is subdivided (deflated) into tiles or parts of tiles from the tile set. Note that the orientations of the new tiles are determined by the markings of the original tile. For P3, the subdivision of thin and thick rhombs is shown in Figures~\ref{fig:penrose}c~and~d.
This deflation can be iterated to produce a patch of any desired size, and defines a tiling of the entire plane in the limit.
A quasiperiodic tiling that can be generated using the deflation and matching technique has the property of being \emph{self-similar}.

\textbf{Generalized dual method (GDM):} This method, discovered by de Bruijn~\cite{deBruijn1981}, creates a tiling from its topological dual.   It starts with $n$ sets of equally spaced parallel lines, called an \emph{$n$-grid} or \emph{multigrid},  and a star of unit vectors, each vector being orthogonal to one of the sets of lines.
For P3, the orientation star, given by $\{\vec{e_i} = (\cos(2\pi i/5), \sin(2\pi i/5))\}_{i=1}^{5}$, is regular (angles are equal) as shown in Figure~\ref{fig:pentagrid}b.
The line sets are in \emph{general position}: that is, at any point only two lines intersect.\footnote{
For details on how to avoid degenerate intersections in the multigrid, we refer the reader to more in-depth discussions~\cite{socolar, boyle, egan}.
}
As in the previous section, we derive a planar graph embedding from the multigrid by placing a vertex at the intersection point of every pair of lines.

The next step is to construct the dual of the multigrid.  Given a primal graph embedding, its \emph{dual} is another graph embedding in which every face in the primal becomes a vertex in the dual, every vertex in the primal becomes a face in the dual, and an edge in the dual is incident to a pair of vertices that correspond to adjacent faces in the primal.
To create a quasiperiodic tiling from a multigrid, the GDM-dual is defined using the following rule:  Within a set of parallel lines, each line is assigned an integer index $k$.
Each face of the multigrid is assigned an $n$-tuple $\textbf{p}{=}(p_1,\cdots,p_{n})$ consisting of ordinal positions, one for each of the $n$ sets of lines.  That is, a face that lies between lines $k$ and $k+1$ in the $i$th set of lines will have $p_i=k$.  The position of a vertex in the dual is then given by $f(\textbf{p}){=}\sum_{i=1}^{n}p_i \vec{e_i}$.

 \begin{center}
 \begin{minipage}{\textwidth}
 \begin{center}
 % width, filename
    \def\svgwidth{0.9\columnwidth}
%% Accompanies image file 'pentagrid.pdf' (pdf, eps, ps)
\begingroup%
  \makeatletter%
  \providecommand\color[2][]{%
    \errmessage{(Inkscape) Color is used for the text in Inkscape, but the package 'color.sty' is not loaded}%
    \renewcommand\color[2][]{}%
  }%
  \providecommand\transparent[1]{%
    \errmessage{(Inkscape) Transparency is used (non-zero) for the text in Inkscape, but the package 'transparent.sty' is not loaded}%
    \renewcommand\transparent[1]{}%
  }%
  \providecommand\rotatebox[2]{#2}%
  \newcommand*\fsize{\dimexpr\f@size pt\relax}%
  \newcommand*\lineheight[1]{\fontsize{\fsize}{#1\fsize}\selectfont}%
  \ifx\svgwidth\undefined%
    \setlength{\unitlength}{983.73937225bp}%
    \ifx\svgscale\undefined%
      \relax%
    \else%
      \setlength{\unitlength}{\unitlength * \real{\svgscale}}%
    \fi%
  \else%
    \setlength{\unitlength}{\svgwidth}%
  \fi%
  \global\let\svgwidth\undefined%
  \global\let\svgscale\undefined%
  \makeatother%
  \begin{picture}(1,0.52982939)%
    \lineheight{1}%
    \setlength\tabcolsep{0pt}%
    \put(0,0){\includegraphics[width=\unitlength,page=1]{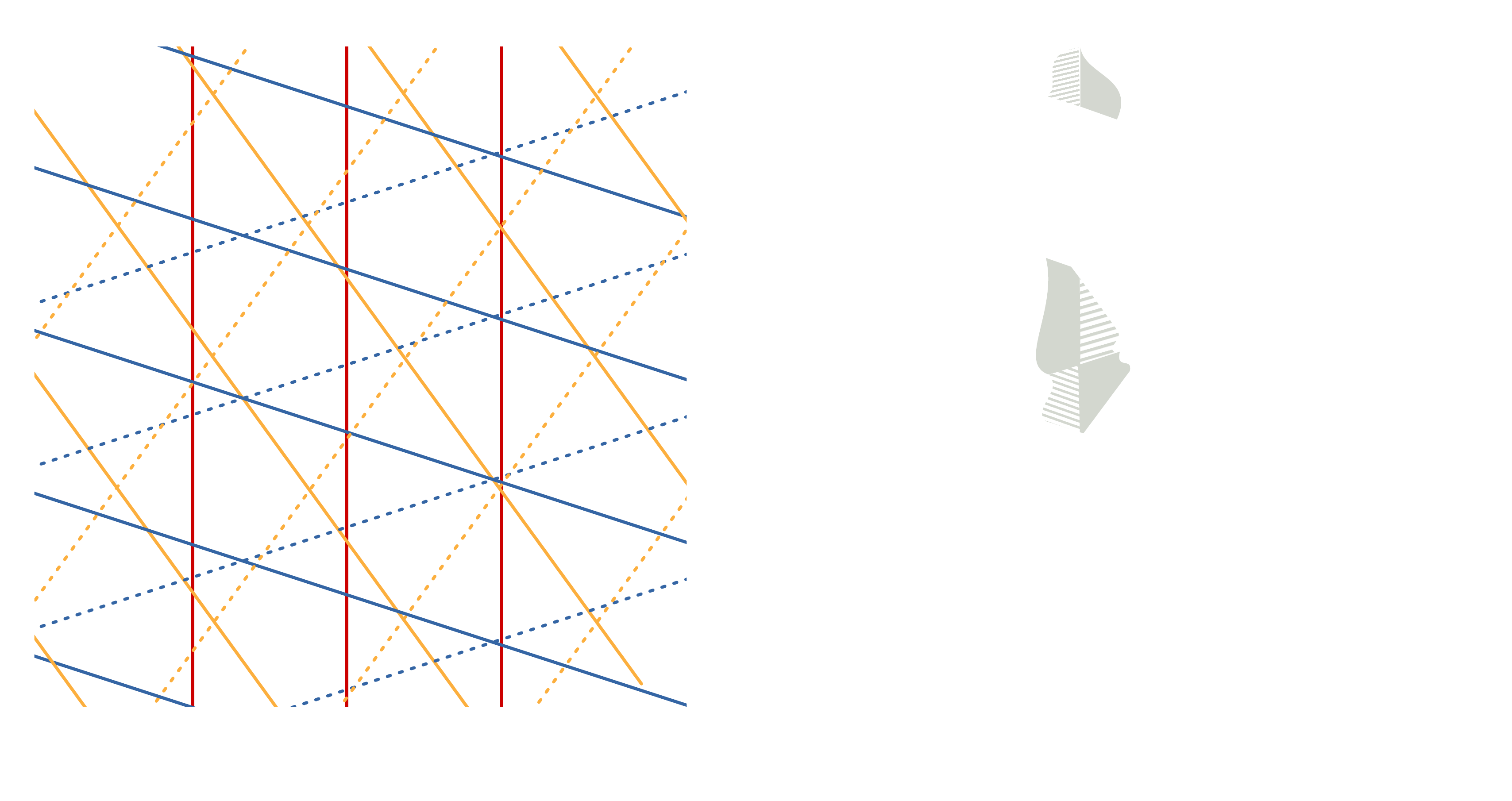}}%
    \put(0.23122347,0.00213927){\color[rgb]{0,0,0}\makebox(0,0)[t]{\lineheight{1.25}\smash{\begin{tabular}[t]{c}\g{(a)}\end{tabular}}}}%
    \put(0.71061123,0.00213927){\color[rgb]{0,0,0}\makebox(0,0)[t]{\lineheight{1.25}\smash{\begin{tabular}[t]{c}\g{(c)}\end{tabular}}}}%
    \put(0.94787684,0.00213927){\color[rgb]{0,0,0}\makebox(0,0)[t]{\lineheight{1.25}\smash{\begin{tabular}[t]{c}\g{(d)}\end{tabular}}}}%
    \put(0,0){\includegraphics[width=\unitlength,page=2]{pentagrid.pdf}}%
    \put(0.50167172,0.15269053){\color[rgb]{0,0,0}\makebox(0,0)[t]{\lineheight{1.25}\smash{\begin{tabular}[t]{c}\g{$\vec{e_5}$}\end{tabular}}}}%
    \put(0.53822439,0.13302425){\color[rgb]{0,0,0}\makebox(0,0)[t]{\lineheight{1.25}\smash{\begin{tabular}[t]{c}\g{$\vec{e_4}$}\end{tabular}}}}%
    \put(0.57869493,0.13314336){\color[rgb]{0,0,0}\makebox(0,0)[t]{\lineheight{1.25}\smash{\begin{tabular}[t]{c}\g{$\vec{e_3}$}\end{tabular}}}}%
    \put(0.62131894,0.15251764){\color[rgb]{0,0,0}\makebox(0,0)[t]{\lineheight{1.25}\smash{\begin{tabular}[t]{c}\g{$\vec{e_2}$}\end{tabular}}}}%
    \put(0.62150203,0.19177189){\color[rgb]{0,0,0}\makebox(0,0)[t]{\lineheight{1.25}\smash{\begin{tabular}[t]{c}\g{$\vec{e_1}$}\end{tabular}}}}%
    \put(0,0){\includegraphics[width=\unitlength,page=3]{pentagrid.pdf}}%
    \put(0.67955312,0.48602939){\color[rgb]{0,0,0}\makebox(0,0)[rt]{\lineheight{1.25}\smash{\begin{tabular}[t]{r}\gt{0}{0}{0}{0}{0}\end{tabular}}}}%
    \put(0.7411124,0.32359328){\color[rgb]{0,0,0}\makebox(0,0)[lt]{\lineheight{1.25}\smash{\begin{tabular}[t]{l}\gt{1}{1}{1}{2}{1}\end{tabular}}}}%
    \put(0.67955312,0.27456701){\color[rgb]{0,0,0}\makebox(0,0)[rt]{\lineheight{1.25}\smash{\begin{tabular}[t]{r}\gt{0}{1}{2}{2}{1}\end{tabular}}}}%
    \put(0.7411124,0.28060375){\color[rgb]{0,0,0}\makebox(0,0)[lt]{\lineheight{1.25}\smash{\begin{tabular}[t]{l}\gt{1}{1}{2}{2}{1}\end{tabular}}}}%
    \put(0.67955312,0.32581161){\color[rgb]{0,0,0}\makebox(0,0)[rt]{\lineheight{1.25}\smash{\begin{tabular}[t]{r}\gt{0}{1}{1}{2}{1}\end{tabular}}}}%
    \put(0.7411124,0.48055654){\color[rgb]{0,0,0}\makebox(0,0)[lt]{\lineheight{1.25}\smash{\begin{tabular}[t]{l}\gt{1}{0}{0}{0}{0}\end{tabular}}}}%
    \put(0,0){\includegraphics[width=\unitlength,page=4]{pentagrid.pdf}}%
    \put(0.87113023,0.51548532){\color[rgb]{0,0,0}\makebox(0,0)[t]{\lineheight{1.25}\smash{\begin{tabular}[t]{c}\gf{0}{0}{0}{0}{0}\end{tabular}}}}%
    \put(0.95299944,0.30956201){\color[rgb]{0,0,0}\makebox(0,0)[lt]{\lineheight{1.25}\smash{\begin{tabular}[t]{l}\gf{1}{1}{1}{2}{1}\end{tabular}}}}%
    \put(0.88856926,0.30640566){\color[rgb]{0,0,0}\makebox(0,0)[rt]{\lineheight{1.25}\smash{\begin{tabular}[t]{r}\gf{0}{1}{2}{2}{1}\end{tabular}}}}%
    \put(0.95604898,0.27009307){\color[rgb]{0,0,0}\makebox(0,0)[lt]{\lineheight{1.25}\smash{\begin{tabular}[t]{l}\gf{1}{1}{2}{2}{1}\end{tabular}}}}%
    \put(0.88704448,0.34715888){\color[rgb]{0,0,0}\makebox(0,0)[rt]{\lineheight{1.25}\smash{\begin{tabular}[t]{r}\gf{0}{1}{1}{2}{1}\end{tabular}}}}%
    \put(0,0){\includegraphics[width=\unitlength,page=5]{pentagrid.pdf}}%
    \put(0.97786592,0.51548532){\color[rgb]{0,0,0}\makebox(0,0)[t]{\lineheight{1.25}\smash{\begin{tabular}[t]{c}\gf{1}{0}{0}{0}{0}\end{tabular}}}}%
    \put(0,0){\includegraphics[width=\unitlength,page=6]{pentagrid.pdf}}%
    \put(0.55508764,0.00213927){\color[rgb]{0,0,0}\makebox(0,0)[t]{\lineheight{1.25}\smash{\begin{tabular}[t]{c}\g{(b)}\end{tabular}}}}%
    \put(0.17848938,0.52029347){\color[rgb]{0.8,0,0}\makebox(0,0)[t]{\lineheight{1.25}\smash{\begin{tabular}[t]{c}\g{$0$}\end{tabular}}}}%
    \put(0.28098234,0.52021903){\color[rgb]{0.8,0,0}\makebox(0,0)[t]{\lineheight{1.25}\smash{\begin{tabular}[t]{c}\g{$1$}\end{tabular}}}}%
    \put(0.39508533,0.52021903){\color[rgb]{0.8,0,0}\makebox(0,0)[t]{\lineheight{1.25}\smash{\begin{tabular}[t]{c}\g{$2$}\end{tabular}}}}%
    \put(0.0657437,0.52021903){\color[rgb]{0.8,0,0}\makebox(0,0)[t]{\lineheight{1.25}\smash{\begin{tabular}[t]{c}\g{$-1$}\end{tabular}}}}%
    \put(0,0){\includegraphics[width=\unitlength,page=7]{pentagrid.pdf}}%
    \put(0.47804692,0.32973556){\color[rgb]{0.20392157,0.39607843,0.64313725}\makebox(0,0)[t]{\lineheight{1.25}\smash{\begin{tabular}[t]{c}\g{1}\end{tabular}}}}%
    \put(0.47797255,0.2222998){\color[rgb]{0.20392157,0.39607843,0.64313725}\makebox(0,0)[t]{\lineheight{1.25}\smash{\begin{tabular}[t]{c}\g{2}\end{tabular}}}}%
    \put(0.47797255,0.11368274){\color[rgb]{0.20392157,0.39607843,0.64313725}\makebox(0,0)[t]{\lineheight{1.25}\smash{\begin{tabular}[t]{c}\g{3}\end{tabular}}}}%
    \put(0.47797255,0.43885273){\color[rgb]{0.20392157,0.39607843,0.64313725}\makebox(0,0)[t]{\lineheight{1.25}\smash{\begin{tabular}[t]{c}\g{0}\end{tabular}}}}%
    \put(0.24932933,0.03540993){\color[rgb]{0,0,0}\makebox(0,0)[t]{\lineheight{1.25}\smash{\begin{tabular}[t]{c}\g{$r$}\end{tabular}}}}%
    \put(0,0){\includegraphics[width=\unitlength,page=8]{pentagrid.pdf}}%
    \put(0.73268903,0.03540993){\color[rgb]{0,0,0}\makebox(0,0)[t]{\lineheight{1.25}\smash{\begin{tabular}[t]{c}\g{$r$}\end{tabular}}}}%
    \put(0,0){\includegraphics[width=\unitlength,page=9]{pentagrid.pdf}}%
  \end{picture}%
\endgroup%

 
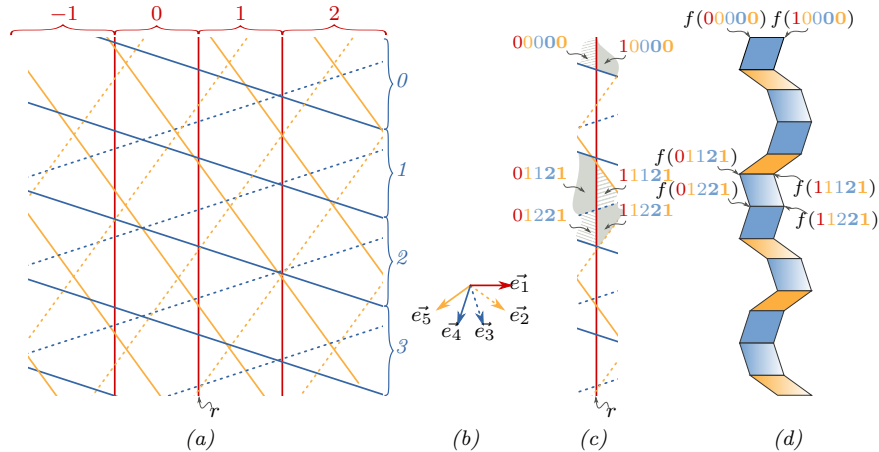
\captionof{figure}[]{Generalized dual method for P3 tiling: a) pentagrid, b) orientation star, c) intersections with a single line, d) corresponding tiles in GDM-dual.}
 \label{fig:pentagrid}
 \end{center}
 \end{minipage}
 \end{center}

%\vfigbegin
%    \imagepdftex{0.9\columnwidth}{pentagrid.pdf_tex}
%\vfigend{fig:pentagrid}{Generalized dual method for P3 tiling: a) pentagrid, b) orientation star, c) intersections with a single line, d) corresponding tiles in GDM-dual.}

The intersection of two lines in the multigrid results in a rhombus in the GDM-dual.
In Figures~\ref{fig:pentagrid}c~and~d, we demonstrate the GDM for P3 by following a single line $r$ in the 5-grid (also known as \emph{pentagrid}).
The result is a sequence of rhombi, each of which has a pair of edges perpendicular to $r$.  This sequence of rhombi is called a `stack' which we define more formally as follows:
\begin{definition}[\cite{deBruijn2013}]
Select an edge $e$ of a parallelogram $p$ in an edge-to-edge tiling by parallelograms.  In the tiling, there are two parallelograms $p'$ and $p''$ that are adjacent to $p$ and have an edge parallel to $e$.  We will call $p'$ and $p''$ the $e$-neighbours of $p$.  A chain of tiles is obtained by taking the $e$-neighbours of $p'$ and $p''$ and iteratively extending this relationship across the tiling.  The resulting bi-infinite sequence of tiles is called a \emph{stack}.
\end{definition}

Applying the GDM to a pentagrid, the set of lines parallel to $r$ will result in a set of stacks, all of which have edges perpendicular to $r$.  This collection of stacks is called a \emph{stack family}.

Two relationships in the GDM-dual are worth noting.  First, the acute and obtuse angles formed by the intersection of two lines in the multigrid correspond to the angles of either a thin rhomb or a thick rhomb in the tiling. Second, in the primal multigrid, only two lines intersect at a point. Thus, in the dual tiling, each tile belongs to exactly two stacks.

Now that we are familiar with the P3 tiling, let us turn to the question of whether it can be a pattern for bobbin lace.
The pentagrid has the right topology for bobbin lace, but its geometry does not satisfy $C0$. The P3 tiling, on the other hand, is well behaved geometrically but its graph structure is not workable. We define a kind of dual that marries the topology of the pentagrid with the geometric arrangement of P3, ultimately producing a pattern that can be rendered in lace.
This dual, which we shall call the \emph{centroid-dual} of the tiling is created by applying the following rule to the P3 tiling: a vertex in the dual graph drawing is placed at the centroid of the corresponding face in the primal tiling.
The centroid of a face is its center of gravity
which can be determined from the positions of the vertices on the boundary of the face by assuming that the face polygon has uniform density.
We will refer to the resulting graph drawing as the \emph{centroid-dual} of the tiling. An example is shown in Figure~\ref{fig:dualP3Proof}c.

To assign a direction to each edge of the centroid-dual, we note that there is a homeomorphism between the centroid-dual and the pentagrid~\cite{deBruijn2013}.  We first orient the lines of the pentagrid by choosing a vector that is not perpendicular to any set of lines and rotating the pentagrid so that this vector points up.  A down direction can be unambiguously assigned to each of the lines in the pentagrid.  Each edge of the centroid-dual is given the same direction as its topological counterpart in the pentagrid.
Note that an edge in the centroid-dual may be horizontal or even angled slightly up but its corresponding edge in the pentagrid always points down.

\begin{theorem}
\label{thm:penrosep3}
The edges of the centroid dual of the P3 tiling can be assigned directions in such a way that the resulting graph satisfies $C0$--$C4$. Therefore, the centroid dual of the P3 tiling is a workable bobbin lace pattern.
\end{theorem}

\begin{proof}
$C0$: All quasiperiodic patterns that can be created by the deflation and matching method, such as the P3 tiling, have vertices that form a Delone point set~\cite{senechal2008}. Therefore, there exists a distance $d$ such that every pair of vertices in P3 is separated by a distance at least $2d$ and a distance $D$ such that every circle of radius $D$ contains at least one point.
The two prototiles of P3, shown in Figure~\ref{thm:penrosep3}, are uniformly bounded.  Therefore there exists a largest incircle of radius $r$ and a smallest containing circle of radius $R$.
P3 satisfies $C0$, but we need to prove that this is also true of the centroid-dual of P3.

First, let us prove that the vertices of the centroid-dual of P3 form a Delone point set.
We observe that the incircle of a rhomb is centered at its centroid.
The incircles of any pair of rhombs in the tiling cannot overlap, therefore, the distance between any pair of vertices in the centroid-dual is at least $2r$.
In the primal graph, a circle of radius $D$ contains at least one vertex $v$.  This vertex is incident to a face $f$ which is bounded by a circle of radius $R$.  The centroid of $f$ is at most a distance $R$ from $v$.  A circle of radius $D+R$ will therefore contain at least one vertex of the centroid-dual.

Now we must prove that the prototiles of the centroid-dual of P3 are uniformly bounded.  De Bruijn showed that there are eight distinct vertex configurations in P3, as shown in Figure~\ref{fig:dualP3C0}a~\cite{deBruijn1981}.  Each face in the centroid-dual maps to a vertex in the primal. The centroid of a rhomb is located at the intersection of its diagonals.  Therefore, we can derive the set of dual prototiles directly from the primal vertex configurations, as shown in Figure~\ref{fig:dualP3C0}b.  We conclude that the centroid-dual tiling is 7-hedral and therefore uniformly bounded.

 \begin{center}
 \begin{minipage}{\textwidth}
 \begin{center}
 % width, filename
    \def\svgwidth{\columnwidth}
%% Accompanies image file 'centroid_dual_p3.pdf' (pdf, eps, ps)
\begingroup%
  \makeatletter%
  \providecommand\color[2][]{%
    \errmessage{(Inkscape) Color is used for the text in Inkscape, but the package 'color.sty' is not loaded}%
    \renewcommand\color[2][]{}%
  }%
  \providecommand\transparent[1]{%
    \errmessage{(Inkscape) Transparency is used (non-zero) for the text in Inkscape, but the package 'transparent.sty' is not loaded}%
    \renewcommand\transparent[1]{}%
  }%
  \providecommand\rotatebox[2]{#2}%
  \newcommand*\fsize{\dimexpr\f@size pt\relax}%
  \newcommand*\lineheight[1]{\fontsize{\fsize}{#1\fsize}\selectfont}%
  \ifx\svgwidth\undefined%
    \setlength{\unitlength}{686.80391693bp}%
    \ifx\svgscale\undefined%
      \relax%
    \else%
      \setlength{\unitlength}{\unitlength * \real{\svgscale}}%
    \fi%
  \else%
    \setlength{\unitlength}{\svgwidth}%
  \fi%
  \global\let\svgwidth\undefined%
  \global\let\svgscale\undefined%
  \makeatother%
  \begin{picture}(1,0.34851947)%
    \lineheight{1}%
    \setlength\tabcolsep{0pt}%
    \put(0,0){\includegraphics[width=\unitlength,page=1]{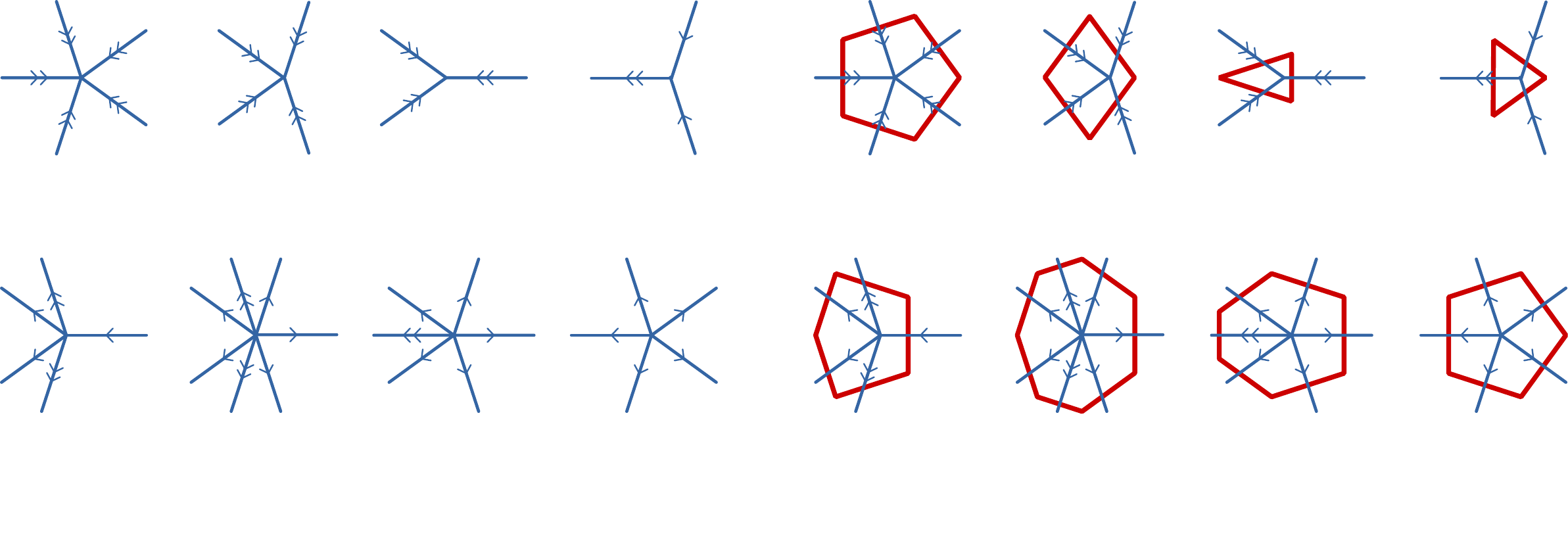}}%
    \put(0.24401248,0.00367703){\color[rgb]{0,0,0}\makebox(0,0)[t]{\lineheight{1.25}\smash{\begin{tabular}[t]{c}\g{(a)}\end{tabular}}}}%
    \put(0.74633912,0.00367703){\color[rgb]{0,0,0}\makebox(0,0)[t]{\lineheight{1.25}\smash{\begin{tabular}[t]{c}\g{(b)}\end{tabular}}}}%
  \end{picture}%
\endgroup%

 \captionof{figure}[]{Demonstrating that the centroid-dual tiling of P3 is 7-hedral: (a) eight vertex configurations of P3, (b) prototiles of centroid-dual in bold red.}
 \label{fig:dualP3C0}
 \end{center}
 \end{minipage}
 \end{center}

$C1$: Every face in the primal graph drawing is a rhomb and therefore every vertex in the dual graph drawing has degree 4.
Each vertex in the centroid-dual is homeomorphic to the intersection of two directed lines in the pentagrid.  Therefore, by homeomorphism, each vertex in the dual graph has in-degree~2 and out-degree~2.

$C2$: The P3 tiling is planar, as is its dual.  Every vertex in an edge-to-edge tiling by rhombs has degree 3 or more, therefore every face in the dual also has degree at~least~3.

$C3$:
The pentagrid was rotated so that all of its line sets are directed downward.  Since it has no upward edges, the pentagrid has no directed cycles.  By homeomorphism, the same can be said of the centroid-dual drawing.

The proof of $C4$ requires a few steps.  First we will prove that a stack in P3 is well-behaved.
A stack is well-behaved if there exists a line $\ell$ such that every tile vertex in the stack is within a finite distance $s$ from $\ell$.
We will then show that in P3, every stack in a stack family contains an osculating path from the centroid-dual of P3 that travels strictly within the stack.  Finally, we will show that the stacks are uniformly distributed and therefore any osculating path in the centroid-dual that is not contained within a stack is sandwiched between well-behaved osculating paths and is therefore also well-behaved.

We start by demonstrating that a stack in P3 is well-behaved.  We phrase the lemma in more general terms for future use.
\begin{lemma}
Let $T$ be a quasiperiodic tiling generated by applying the generalized dual method to a multigrid consisting of $n$ sets of equally spaced lines with a regular orientation star.  The stacks of $T$ are well-behaved.
\label{thm:wellstack}
\end{lemma}
\begin{proof}
Select a line $k$ from the multigrid.  For ease of explanation, let us assume that $k$ is a vertical line belonging to $A_0$, the set of vertical lines.  See, for example, the vertical red line in Figure~\ref{fig:pentagrid}c.
Because the orientation star is regular, the remaining sets of parallel lines can be paired up with set $A_i$ being the set that intersects $k$ at angle $\theta_i$ and $A_i'$ being the set that intersects line $k$ at angle $-\theta_i$, for $1 \le i \le \floor{n/2}$.
When $n$ is even, one set, $A_{\ceil{n/2}}$ will be horizontal.

The translational spacing between lines within the set is the same for all sets, therefore, the distance between the points where lines in set $A_i$ intersect line $k$ is equal to the distance between the points where lines of set $A_i'$ intersect line $k$.
The subsequence formed by the intersections of $A_i$ and $A_i'$ with line $k$ must therefore alternate  (one set of intersections cannot `catch up and pass' the other).  This alternation between lines with opposite angle is known as the alternation condition and applies for all values of $i$.
The rhombs in the GDM-dual resulting from an intersection of line $k$ and any line from $A_i$ tilt to the left of the line $k$, making the stack bend to the left.  Rhombs resulting from an intersection of line $k$ with a line from $A_i'$ tilt to the right with the same magnitude, making the stack bend an equal amount to the right.  Because the rhombs in the stack alternate between $A_i$ and $A_i'$, the amount of bend left or right is balanced making the stack well-behaved.
When the number of non-vertical sets is odd, the orthogonal intersection of a line from $A_{\ceil{n/2}}$ with line $k$ corresponds to a square that is edge aligned with line $k$ and therefore has no bending effect on the stack.
\end{proof}

Now we will show that each stack in P3 contains an osculating path and, by inclusion, the path is therefore also well-behaved.
\begin{lemma}
Let $T$ be a P3 tiling and let $T'$ be a deflation of $T$ according to the deflation and matching method.
For each stack $s$ of $T$, there exists a path in the osculating partition of the centroid-dual of $T'$ that does not cross the boundary edges of $s$.
\end{lemma}
\begin{proof}
Proof by exhaustion: We will consider each of the tile configurations that can occur in a stack of $T$ (up to rotation by $\pi$, vertical or horizontal reflection) and demonstrate that each configuration includes a subpath of an osculating path in the centroid-dual of $T'$. Further, when two configurations are joined together, their osculating subpaths connect.

Once again, from de Bruijn's extensive analysis of P3, we know that there are seven arrangements of tiles to consider (see Figure~\ref{fig:dualP3Proof}a).
\begin{lemma}\cite{deBruijn2013}
A \emph{central tile configuration} is a central tile and its four direct neighbours, each neighbour having one edge in common with the central tile.   For P3, there exist seven distinct central tile configurations up to rotation.  Further, the matching conditions can be applied to each configuration in only one way.
\end{lemma}

Let $t$ be the central tile of a configuration.  For convenience of discussion, we will assume that the edge $e$ which defines the stack $s$ is horizontal. In Figure~\ref{fig:dualP3Proof}a, we have highlighted one stack through the central tile of each central tile configuration in white.  Every tile belongs to two stacks; however, choosing the alternate stack $s'$ in each configuration and orienting $s'$ so the defining edge is horizontal will result in the same set of drawings up to reflection in a vertical mirror.

Within $t$, select the edges of the centroid-dual that extend from a horizontal edge of $t$ to a non-horizontal edge of $t$ or vice versa, see Figure~\ref{fig:dualP3Proof}d. We observe that in all seven central tile configurations these edges exist and connect to form a direct path from the top of $t$ to the bottom of $t$.  Further, the tile preceding $t$ in $s$ and the tile following $t$ in $s$ also have such a path and these three paths are all connected.
The path thus formed, contained within the boundary of $s$, does not transversely cross a pair of edges in the centroid-dual drawing, so it is a subpath of an osculating path.  Under rotation by $\pi$ or reflection in the horizontal or vertical direction, the selected edges remain unchanged.
The continuation in preceding and succeeding tiles proves that the subpaths connect to form an infinite path.

\vfigbegin
    % width, filename
    \def\svgwidth{\columnwidth}
%% Accompanies image file 'dualP3_proof3.pdf' (pdf, eps, ps)
\begingroup%
  \makeatletter%
  \providecommand\color[2][]{%
    \errmessage{(Inkscape) Color is used for the text in Inkscape, but the package 'color.sty' is not loaded}%
    \renewcommand\color[2][]{}%
  }%
  \providecommand\transparent[1]{%
    \errmessage{(Inkscape) Transparency is used (non-zero) for the text in Inkscape, but the package 'transparent.sty' is not loaded}%
    \renewcommand\transparent[1]{}%
  }%
  \providecommand\rotatebox[2]{#2}%
  \newcommand*\fsize{\dimexpr\f@size pt\relax}%
  \newcommand*\lineheight[1]{\fontsize{\fsize}{#1\fsize}\selectfont}%
  \ifx\svgwidth\undefined%
    \setlength{\unitlength}{646.49061584bp}%
    \ifx\svgscale\undefined%
      \relax%
    \else%
      \setlength{\unitlength}{\unitlength * \real{\svgscale}}%
    \fi%
  \else%
    \setlength{\unitlength}{\svgwidth}%
  \fi%
  \global\let\svgwidth\undefined%
  \global\let\svgscale\undefined%
  \makeatother%
  \begin{picture}(1,1.1925576)%
    \lineheight{1}%
    \setlength\tabcolsep{0pt}%
    \put(0,0){\includegraphics[width=\unitlength,page=1]{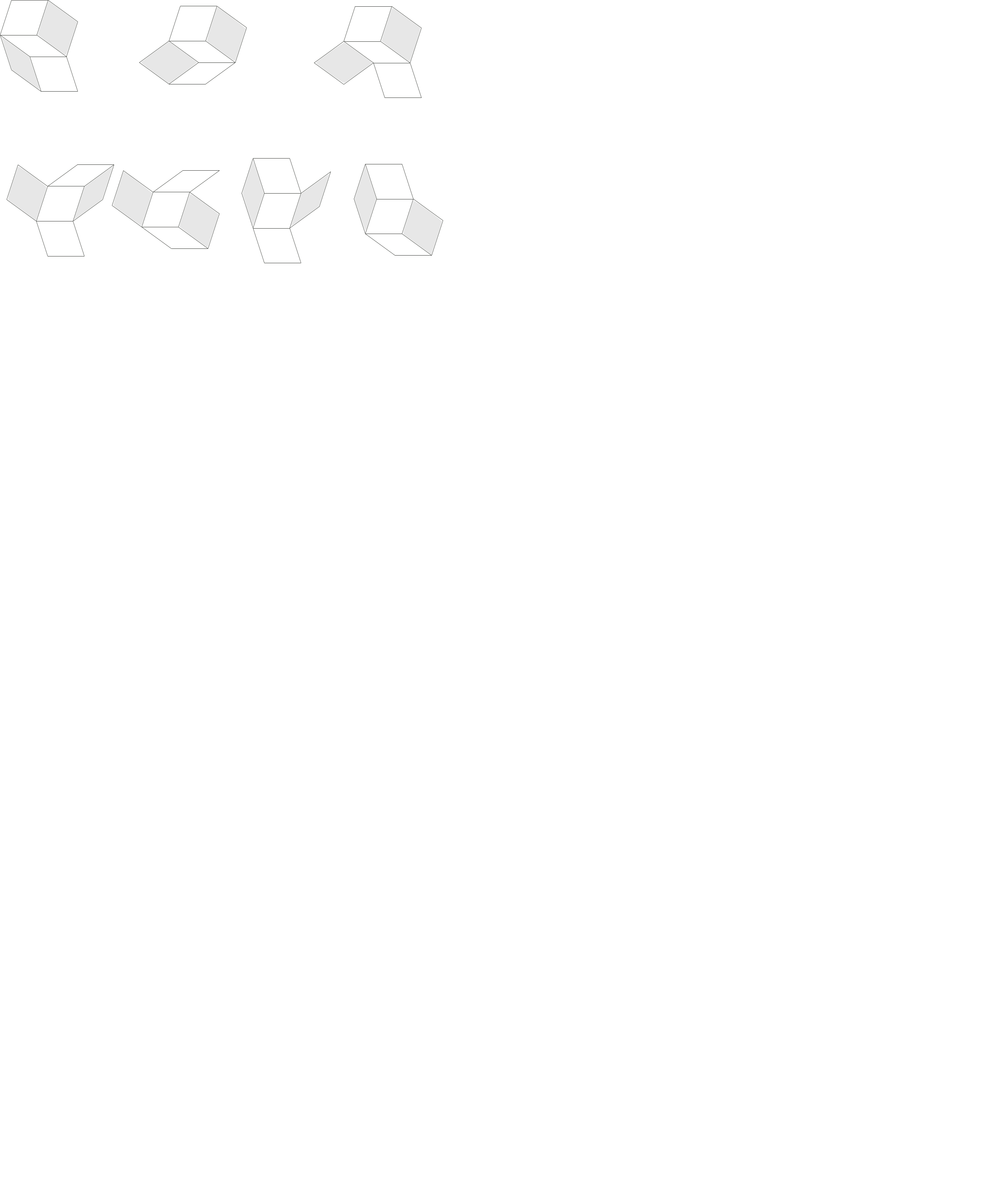}}%
    \put(0.23691031,0.87575595){\color[rgb]{0,0,0}\makebox(0,0)[t]{\lineheight{1.25}\smash{\begin{tabular}[t]{c}\g{(a)}\end{tabular}}}}%
    \put(0.79535695,0.01889007){\color[rgb]{0,0,0}\makebox(0,0)[t]{\lineheight{1.25}\smash{\begin{tabular}[t]{c}\g{(d)}\end{tabular}}}}%
    \put(0,0){\includegraphics[width=\unitlength,page=2]{dualP3_proof3.pdf}}%
    \put(0.23690585,0.00393755){\color[rgb]{0,0,0}\makebox(0,0)[t]{\lineheight{1.25}\smash{\begin{tabular}[t]{c}\g{(e)}\end{tabular}}}}%
    \put(0.04029573,1.06644359){\color[rgb]{0,0,0}\makebox(0,0)[lt]{\lineheight{1.25}\smash{\begin{tabular}[t]{l}\g{$dB1$}\end{tabular}}}}%
    \put(0.35432715,1.06644359){\color[rgb]{0,0,0}\makebox(0,0)[lt]{\lineheight{1.25}\smash{\begin{tabular}[t]{l}\g{$dB3$}\end{tabular}}}}%
    \put(0.19792092,1.06644359){\color[rgb]{0,0,0}\makebox(0,0)[lt]{\lineheight{1.25}\smash{\begin{tabular}[t]{l}\g{$dB2$}\end{tabular}}}}%
    \put(0.3960844,0.90806554){\color[rgb]{0,0,0}\makebox(0,0)[lt]{\lineheight{1.25}\smash{\begin{tabular}[t]{l}\g{$dB7$}\end{tabular}}}}%
    \put(0.05256457,0.90846206){\color[rgb]{0,0,0}\makebox(0,0)[lt]{\lineheight{1.25}\smash{\begin{tabular}[t]{l}\g{$dB4$}\end{tabular}}}}%
    \put(0.1644752,0.90846206){\color[rgb]{0,0,0}\makebox(0,0)[lt]{\lineheight{1.25}\smash{\begin{tabular}[t]{l}\g{$dB5$}\end{tabular}}}}%
    \put(0.28417363,0.90806554){\color[rgb]{0,0,0}\makebox(0,0)[lt]{\lineheight{1.25}\smash{\begin{tabular}[t]{l}\g{$dB6$}\end{tabular}}}}%
    \put(0,0){\includegraphics[width=\unitlength,page=3]{dualP3_proof3.pdf}}%
    \put(0.19792092,0.54150103){\color[rgb]{0,0,0}\makebox(0,0)[lt]{\lineheight{1.25}\smash{\begin{tabular}[t]{l}\g{$dB1$}\end{tabular}}}}%
    \put(0,0){\includegraphics[width=\unitlength,page=4]{dualP3_proof3.pdf}}%
    \put(0.11924249,0.5092049){\color[rgb]{0,0,0}\makebox(0,0)[t]{\lineheight{1.25}\smash{\begin{tabular}[t]{c}\g{(b)}\end{tabular}}}}%
    \put(0.40279484,0.51306815){\color[rgb]{0,0,0}\makebox(0,0)[t]{\lineheight{1.25}\smash{\begin{tabular}[t]{c}\g{(c)}\end{tabular}}}}%
    \put(0,0){\includegraphics[width=\unitlength,page=5]{dualP3_proof3.pdf}}%
    \put(0.64490244,0.96123236){\color[rgb]{0,0,0}\makebox(0,0)[lt]{\lineheight{1.25}\smash{\begin{tabular}[t]{l}\g{$dB1$}\end{tabular}}}}%
    \put(0,0){\includegraphics[width=\unitlength,page=6]{dualP3_proof3.pdf}}%
    \put(0.64490246,0.68033346){\color[rgb]{0,0,0}\makebox(0,0)[lt]{\lineheight{1.25}\smash{\begin{tabular}[t]{l}\g{$dB3$}\end{tabular}}}}%
    \put(0,0){\includegraphics[width=\unitlength,page=7]{dualP3_proof3.pdf}}%
    \put(0.86893075,0.82278493){\color[rgb]{0,0,0}\makebox(0,0)[lt]{\lineheight{1.25}\smash{\begin{tabular}[t]{l}\g{$dB2$}\end{tabular}}}}%
    \put(0,0){\includegraphics[width=\unitlength,page=8]{dualP3_proof3.pdf}}%
    \put(0.6449024,0.10159449){\color[rgb]{0,0,0}\makebox(0,0)[lt]{\lineheight{1.25}\smash{\begin{tabular}[t]{l}\g{$dB7$}\end{tabular}}}}%
    \put(0,0){\includegraphics[width=\unitlength,page=9]{dualP3_proof3.pdf}}%
    \put(0.86143217,0.53200781){\color[rgb]{0,0,0}\makebox(0,0)[lt]{\lineheight{1.25}\smash{\begin{tabular}[t]{l}\g{$dB4$}\end{tabular}}}}%
    \put(0,0){\includegraphics[width=\unitlength,page=10]{dualP3_proof3.pdf}}%
    \put(0.6449024,0.39370637){\color[rgb]{0,0,0}\makebox(0,0)[lt]{\lineheight{1.25}\smash{\begin{tabular}[t]{l}\g{$dB5$}\end{tabular}}}}%
    \put(0,0){\includegraphics[width=\unitlength,page=11]{dualP3_proof3.pdf}}%
    \put(0.86143218,0.24041235){\color[rgb]{0,0,0}\makebox(0,0)[lt]{\lineheight{1.25}\smash{\begin{tabular}[t]{l}\g{$dB6$}\end{tabular}}}}%
  \end{picture}%
\endgroup%

\vfigend{fig:dualP3Proof}{Proving of $C4$ for centroid-dual of P3, white tiles indicate stacks: a) 7 central tile configurations of P3 tiling $T$, b) Deflation of configurations in (a) c) centroid-dual of (b) d) overlay of (a) and (c) e) a patch of P3 overlaid with osculating partition of centroid-dual. }{Proof of C4 for centroid-dual of P3}
\end{proof}

Finally we show that the stacks of P3 are uniformly distributed throughout the tiling
and therefore there is also a uniform distribution of well-behaved paths from the osculating partition of the centroid-dual.
\begin{lemma}
\label{thm:uniformstack}
In P3, the stacks of a stack family are uniformly distributed throughout the tiling.
\end{lemma}
\begin{proof}
Let $S_e$ be the stack family defined by all edges parallel to $e$.
Select a tile $t$ from one stack $s$ in $S_e$.  A second stack $r$ defined by edge $f$ also passes through $t$ and intersects not only with $s$ but with all stacks $S_e$.
The tiles at these intersection points have two sides parallel to $e$ and two sides parallel to $f$ and therefore are all translations of $t$ along $r$.  We know from the generalized dual method that these congruent tiles
are dual to a periodic set of lines intersecting with the line dual to $r$.  They are therefore uniformly distributed along $r$.  We can therefore conclude that the stacks of $S_e$ are uniformly distributed in the tiling.
\end{proof}

By Lemmas~\ref{thm:wellstack}~and~\ref{thm:uniformstack}, we have shown that there exists a set $W$ of uniformly distributed, well-behaved osculating paths in the oriented centroid-dual of the P3 tiling.  The set $W$ is a subset of the osculating partition of the tiling.  However, any osculating path not in $W$ is sandwiched between two paths that are in $W$ and must therefore also be well-behaved.  This concludes the proof of $C4$.
\end{proof}

We have demonstrated that the P3 tiling gives rise to a lace pattern.  Based on empirical exploration, we conjecture that other self-similar quasiperiodic tilings, such as the Ammann-Beenker tiling used to create the lace in Figure~\ref{fig:lace2}, will also produce lace patterns.

\begin{conjecture}
Let $T$ be a quasiperiodic tiling by edge-to-edge rhombs created by applying the generalized dual method to sets of periodically spaced parallel lines and let $T'$ be the centroid-dual of $T$.  For all $T$ there exists a set of edge directions of $T'$ that satisfies $C0$--$C4$  and thus forms a workable bobbin lace pattern.
\end{conjecture}

\section{Lace pattern from the Ammann bar decoration of P3}
There are several ways in which the P3 tiles can be decorated to enforce the local matching rules.  One of the most well known is the arcs proposed by Conway which form attractive curved designs with local fivefold symmetry.  Another significant decoration is the  Ammann bars, shown in Figure~\ref{fig:ammann}a. The tiles are decorated by line segments and the local matching rule requires that adjacent tiles extend the segments, without bending, to form an infinite line in the limit.
Ammann bars were discovered by several people in the late 1970s~\cite{grunbaum} and named after Ammann because he was
the first to recognize that they clearly illustrate how the matching rules affect not just adjacent tiles but also tiles at a distance.

As with the pentagrid, the Ammann bars fall into five regularly-spaced families of parallel lines in general position.
The orientation star of the sets is a regular spacing of $\pi i /5$ and the lines are in general position.  A key difference between the Ammann bars and the pentagrid is that the spacing between lines is not periodic but rather a sequence of long and short spaces corresponding to a Fibonacci word.  Whereas the faces of the Pentagrid can be vanishingly small and belong to an infinite set of shapes, Ammann bars divide the plane into a finite set of distinct tile shapes.

We can create a graph drawing from the Ammann bars, just as we have done in previous examples, by placing a vertex at every line intersection and an edge between consecutive pairs of vertices along a line. A set of edge directions is assigned by choosing an up vector, rotating the tiling until no line in the Ammann bars is horizontal and directing all lines downward.

\begin{theorem}
There exists a set of directions for the edges of the graph drawing derived from the Ammann bar decoration of the P3 tiling that satisfies conditions $C0$--$C4$  and thus forms a workable bobbin lace pattern.
\end{theorem}
\begin{proof}

$C0$: A patch of the Ammann bar decoration for P3 can be obtained via iterative substitution, independent of the tiling it decorates.  It is self-similar, therefore, the vertices form a Delone point set~\cite{senechal}. The complement of the Ammann bars divides the plane into cells from a finite set of prototiles from which we can observe that the faces are uniformly bounded.

$C1$: Each vertex is the intersection of two directed lines, therefore, it has in-degree~2 and out-degree~2.

$C2$: A vertex is placed at every line crossing so the graph embedding is planar. Further, the straight lines prevent self-loop or bigons.

$C3$: Since the graph drawing has no upward edges it has no directed cycles.

$C4$: As shown in Figure~\ref{fig:ammann}, we again exhaustively examine each of the seven de Bruijn configurations and show that every stack in the P3 tiling contains an osculating path from the Ammann bar decoration.  Using Lemmas~\ref{thm:wellstack}~and~\ref{thm:uniformstack}, that the rhombs of the P3 tiling form well-behaved stacks that are uniformly distributed throughout the tiling, we can conclude that an osculating partition of the Ammann bars of P3 contains a uniform distribution of well-behaved paths.  Because  the remaining paths are sandwiched between well-behaved paths, all paths in the partition are well-behaved.

\vfigbegin
    % width, filename
    \def\svgwidth{\columnwidth}
%% Accompanies image file 'ammann2.pdf' (pdf, eps, ps)
\begingroup%
  \makeatletter%
  \providecommand\color[2][]{%
    \errmessage{(Inkscape) Color is used for the text in Inkscape, but the package 'color.sty' is not loaded}%
    \renewcommand\color[2][]{}%
  }%
  \providecommand\transparent[1]{%
    \errmessage{(Inkscape) Transparency is used (non-zero) for the text in Inkscape, but the package 'transparent.sty' is not loaded}%
    \renewcommand\transparent[1]{}%
  }%
  \providecommand\rotatebox[2]{#2}%
  \newcommand*\fsize{\dimexpr\f@size pt\relax}%
  \newcommand*\lineheight[1]{\fontsize{\fsize}{#1\fsize}\selectfont}%
  \ifx\svgwidth\undefined%
    \setlength{\unitlength}{504bp}%
    \ifx\svgscale\undefined%
      \relax%
    \else%
      \setlength{\unitlength}{\unitlength * \real{\svgscale}}%
    \fi%
  \else%
    \setlength{\unitlength}{\svgwidth}%
  \fi%
  \global\let\svgwidth\undefined%
  \global\let\svgscale\undefined%
  \makeatother%
  \begin{picture}(1,0.5977414)%
    \lineheight{1}%
    \setlength\tabcolsep{0pt}%
    \put(0,0){\includegraphics[width=\unitlength,page=1]{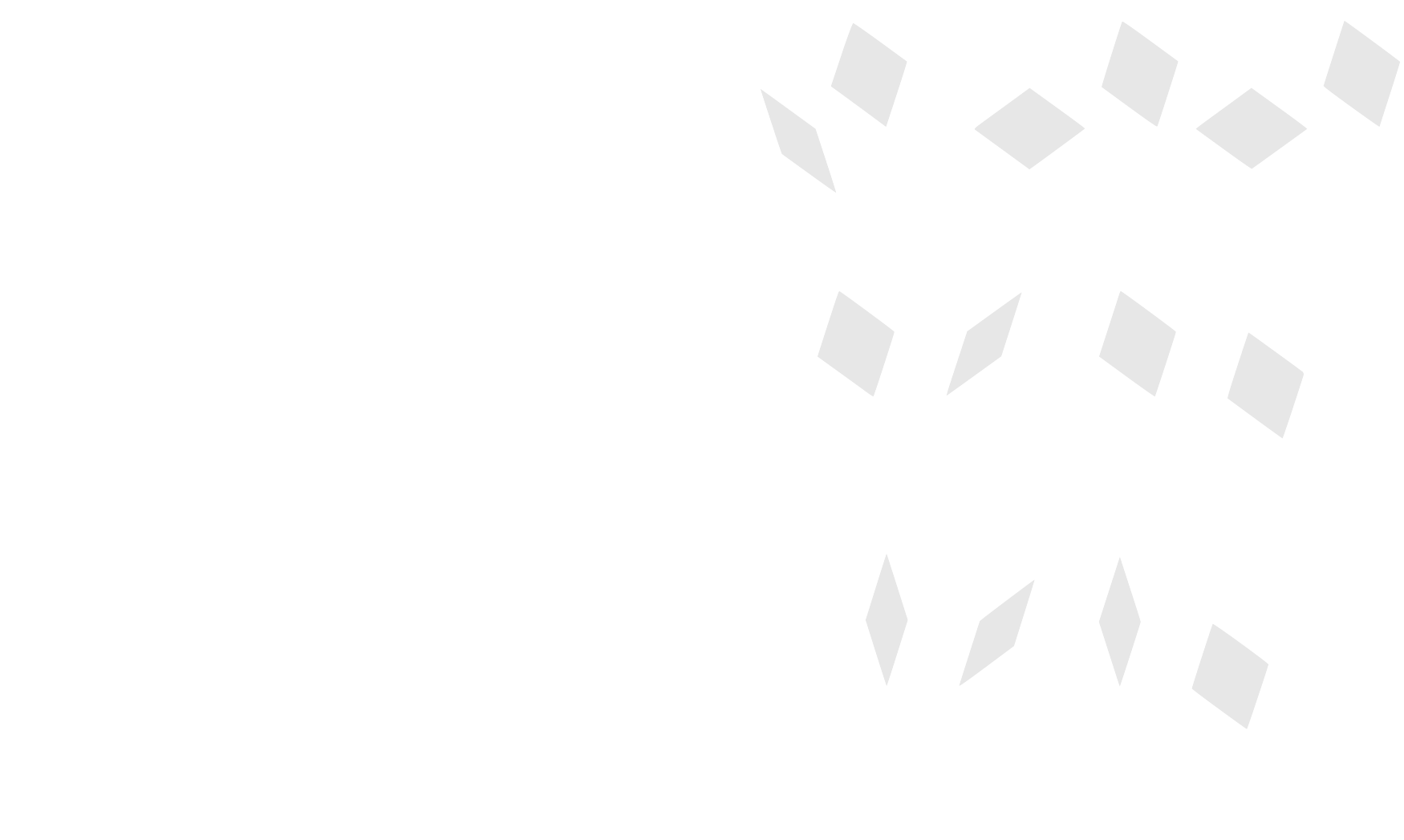}}%
    \put(0.25727173,0.0014805){\color[rgb]{0,0,0}\makebox(0,0)[lt]{\lineheight{0}\smash{\begin{tabular}[t]{l}\g{(b)}\end{tabular}}}}%
    \put(0.76063459,0.00325577){\color[rgb]{0,0,0}\makebox(0,0)[lt]{\lineheight{0}\smash{\begin{tabular}[t]{l}\g{(c)}\end{tabular}}}}%
    \put(0,0){\includegraphics[width=\unitlength,page=2]{ammann2.pdf}}%
    \put(0.76131918,0.42658882){\color[rgb]{0,0,0}\makebox(0,0)[lt]{\lineheight{0}\smash{\begin{tabular}[t]{l}\g{$dB2$}\end{tabular}}}}%
    \put(0,0){\includegraphics[width=\unitlength,page=3]{ammann2.pdf}}%
    \put(0.57667783,0.42658882){\color[rgb]{0,0,0}\makebox(0,0)[lt]{\lineheight{0}\smash{\begin{tabular}[t]{l}\g{$dB1$}\end{tabular}}}}%
    \put(0,0){\includegraphics[width=\unitlength,page=4]{ammann2.pdf}}%
    \put(0.91623725,0.42658882){\color[rgb]{0,0,0}\makebox(0,0)[lt]{\lineheight{0}\smash{\begin{tabular}[t]{l}\g{$dB3$}\end{tabular}}}}%
    \put(0,0){\includegraphics[width=\unitlength,page=5]{ammann2.pdf}}%
    \put(0.85028899,0.22819796){\color[rgb]{0,0,0}\makebox(0,0)[lt]{\lineheight{0}\smash{\begin{tabular}[t]{l}\g{$dB5$}\end{tabular}}}}%
    \put(0,0){\includegraphics[width=\unitlength,page=6]{ammann2.pdf}}%
    \put(0.65791788,0.22819796){\color[rgb]{0,0,0}\makebox(0,0)[lt]{\lineheight{0}\smash{\begin{tabular}[t]{l}\g{$dB4$}\end{tabular}}}}%
    \put(0,0){\includegraphics[width=\unitlength,page=7]{ammann2.pdf}}%
    \put(0.81109561,0.02619219){\color[rgb]{0,0,0}\makebox(0,0)[lt]{\lineheight{0}\smash{\begin{tabular}[t]{l}\g{$dB7$}\end{tabular}}}}%
    \put(0,0){\includegraphics[width=\unitlength,page=8]{ammann2.pdf}}%
    \put(0.6678736,0.02103681){\color[rgb]{0,0,0}\makebox(0,0)[lt]{\lineheight{0}\smash{\begin{tabular}[t]{l}\g{$dB6$}\end{tabular}}}}%
    \put(0,0){\includegraphics[width=\unitlength,page=9]{ammann2.pdf}}%
    \put(0.25727173,0.5099678){\color[rgb]{0,0,0}\makebox(0,0)[lt]{\lineheight{0}\smash{\begin{tabular}[t]{l}\g{(a)}\end{tabular}}}}%
    \put(0,0){\includegraphics[width=\unitlength,page=10]{ammann2.pdf}}%
  \end{picture}%
\endgroup%

\vfigend{fig:ammann}{Ammann bars for P3: a) prototiles with bar decorations, b) a patch of P3, c) a path from the osculating partition appears in each of de Bruijn's central tile configurations.}{Ammann bars}

\end{proof}

The Ammann bar decoration of P3 can be generalized to decorate a family of quasiperiodic tilings.  The defining characteristic of these decorations is that the line segments in adjacent tiles align to form straight lines in the infinite tiling, lines which we refer to as \emph{Ammann lines}.  Boyle and Steinhardt give a detailed analysis of these patterns~\cite{boyle} and use this generalization to generate new tilings with matching rules.  We note that some of their Ammann line patterns, such as the new 12-fold-B and 12-fold-C patterns, are not in general position.

\begin{conjecture}
For every set of Ammann lines in general position there exists a set of edge directions for the associated graph drawing that satisfies $C0$--$C4$.  Therefore Ammann lines form a workable bobbin lace pattern.
\end{conjecture}

\section{Assigning a braid word mapping}
As mentioned in Section~\ref{sec:problem}, a lace pattern consists of two elements: a graph drawing and a mapping from each vertex to a braid word.
Deciding which braid word to use at each vertex is an artistic choice: any combination of crosses and twists containing at least one cross will work.

In a periodic pattern, we would typically select a unit cell and assign a braid word to each vertex in that parallelogram, paying attention to any symmetry elements that we wish to emphasize and considering how symmetrically related vertices are positioned in space.\footnote{When one vertex is rotated relative to another, the edge directions around the two vertices may be quite different.  It may not be possible to construct two braid words that look sufficiently similar in both orientations.
}
The same braid word mapping would then be applied to all translations of the unit cell.  We can therefore describe the mapping of an unbounded number of vertices with a finite number of braid words.

In a quasiperiodic pattern, there is no unit cell.  How does this affect the assignment of braid words?
Fortunately, for the patterns explored here, there is a repeating unit called a quasi-unit cell~\cite{steinhardt1999}. Roughly speaking, a unit cell tiles the plane without gaps or overlaps whereas copies of a quasi-unit cell are allowed to overlap, creating a covering (as opposed to a tiling) of the plane. So once again we are able to consider a single subpatch of the pattern, make some artistic choices and apply them to the entire pattern.
The quasi-unit cell has a finite number of vertices  but, unlike the unit cell, it can appear in a finite number of distinct orientations, all of which must be considered.  To simplify the process, we  first look at the number of twists that we want to appear in two threads travelling between braids, label the edges of the drawing accordingly, and use the edge labels to determine the braid word at each vertex for each rotation.  This process is illustrated in Figure~\ref{fig:design} for the lace appearing in Figure~\ref{fig:lace2}.

\vfigbegin
    % width, filename
    \def\svgwidth{\columnwidth}
%% Accompanies image file 'design2.pdf' (pdf, eps, ps)
\begingroup%
  \makeatletter%
  \providecommand\color[2][]{%
    \errmessage{(Inkscape) Color is used for the text in Inkscape, but the package 'color.sty' is not loaded}%
    \renewcommand\color[2][]{}%
  }%
  \providecommand\transparent[1]{%
    \errmessage{(Inkscape) Transparency is used (non-zero) for the text in Inkscape, but the package 'transparent.sty' is not loaded}%
    \renewcommand\transparent[1]{}%
  }%
  \providecommand\rotatebox[2]{#2}%
  \newcommand*\fsize{\dimexpr\f@size pt\relax}%
  \newcommand*\lineheight[1]{\fontsize{\fsize}{#1\fsize}\selectfont}%
  \ifx\svgwidth\undefined%
    \setlength{\unitlength}{448.32279968bp}%
    \ifx\svgscale\undefined%
      \relax%
    \else%
      \setlength{\unitlength}{\unitlength * \real{\svgscale}}%
    \fi%
  \else%
    \setlength{\unitlength}{\svgwidth}%
  \fi%
  \global\let\svgwidth\undefined%
  \global\let\svgscale\undefined%
  \makeatother%
  \begin{picture}(1,0.23342817)%
    \lineheight{1}%
    \setlength\tabcolsep{0pt}%
    \put(0,0){\includegraphics[width=\unitlength,page=1]{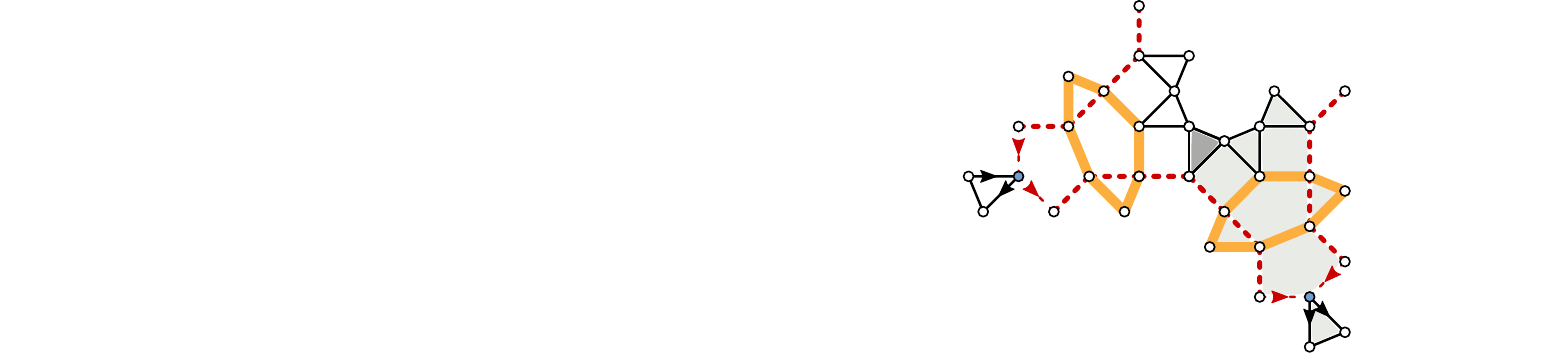}}%
    \put(0.75523305,0.03081032){\color[rgb]{0,0,0}\makebox(0,0)[t]{\lineheight{1.25}\smash{\begin{tabular}[t]{c}\g{CTCT}\end{tabular}}}}%
    \put(0,0){\includegraphics[width=\unitlength,page=2]{design2.pdf}}%
    \put(0.62734618,0.18091287){\color[rgb]{0,0,0}\makebox(0,0)[t]{\lineheight{1.25}\smash{\begin{tabular}[t]{c}\g{CTCL}\end{tabular}}}}%
    \put(0,0){\includegraphics[width=\unitlength,page=3]{design2.pdf}}%
    \put(0.85272497,0.195819){\color[rgb]{0,0,0}\makebox(0,0)[t]{\lineheight{1.25}\smash{\begin{tabular}[t]{c}\g{overlap}\end{tabular}}}}%
    \put(0,0){\includegraphics[width=\unitlength,page=4]{design2.pdf}}%
    \put(0.45934304,0.16976012){\color[rgb]{0,0,0}\makebox(0,0)[lt]{\lineheight{1.25}\smash{\begin{tabular}[t]{l}\g{1 twist}\end{tabular}}}}%
    \put(0,0){\includegraphics[width=\unitlength,page=5]{design2.pdf}}%
    \put(0.46178615,0.11301296){\color[rgb]{0,0,0}\makebox(0,0)[lt]{\lineheight{1.25}\smash{\begin{tabular}[t]{l}\g{no twist}\end{tabular}}}}%
    \put(0.46393233,0.05908222){\color[rgb]{0,0,0}\makebox(0,0)[lt]{\lineheight{1.25}\smash{\begin{tabular}[t]{l}\g{grand Venetian cord}\end{tabular}}}}%
    \put(0,0){\includegraphics[width=\unitlength,page=6]{design2.pdf}}%
    \put(0.29139724,0.00375532){\color[rgb]{0,0,0}\makebox(0,0)[t]{\lineheight{1.25}\smash{\begin{tabular}[t]{c}\g{(b)}\end{tabular}}}}%
    \put(0.09407944,0.00375532){\color[rgb]{0,0,0}\makebox(0,0)[t]{\lineheight{1.25}\smash{\begin{tabular}[t]{c}\g{(a)}\end{tabular}}}}%
    \put(0.6817048,0.00375532){\color[rgb]{0,0,0}\makebox(0,0)[t]{\lineheight{1.25}\smash{\begin{tabular}[t]{c}\g{(c)}\end{tabular}}}}%
    \put(0,0){\includegraphics[width=\unitlength,page=7]{design2.pdf}}%
  \end{picture}%
\endgroup%

\vfigend{fig:design}{Quasi-unit cell and braid word considerations for `Nodding bur-marigold': a) quasi-unit cell overlayed on Ammann-Beenker tiles, b) edges labelled with twist information, c) three copies of the quasi-unit cell showing the overlap of a face and two different braid words assigned to the same vertex of the quasi-unit cell based on rotation.}{Design}

\section{Artistic results and conclusion}
\vfigbegin
    % width, filename
    \def\svgwidth{\columnwidth}
%% Accompanies image file 'lace_p3_bars.pdf' (pdf, eps, ps)
\begingroup%
  \makeatletter%
  \providecommand\color[2][]{%
    \errmessage{(Inkscape) Color is used for the text in Inkscape, but the package 'color.sty' is not loaded}%
    \renewcommand\color[2][]{}%
  }%
  \providecommand\transparent[1]{%
    \errmessage{(Inkscape) Transparency is used (non-zero) for the text in Inkscape, but the package 'transparent.sty' is not loaded}%
    \renewcommand\transparent[1]{}%
  }%
  \providecommand\rotatebox[2]{#2}%
  \newcommand*\fsize{\dimexpr\f@size pt\relax}%
  \newcommand*\lineheight[1]{\fontsize{\fsize}{#1\fsize}\selectfont}%
  \ifx\svgwidth\undefined%
    \setlength{\unitlength}{194.87998581bp}%
    \ifx\svgscale\undefined%
      \relax%
    \else%
      \setlength{\unitlength}{\unitlength * \real{\svgscale}}%
    \fi%
  \else%
    \setlength{\unitlength}{\svgwidth}%
  \fi%
  \global\let\svgwidth\undefined%
  \global\let\svgscale\undefined%
  \makeatother%
  \begin{picture}(1,0.50431043)%
    \lineheight{1}%
    \setlength\tabcolsep{0pt}%
    \put(0,0){\includegraphics[width=\unitlength,page=1]{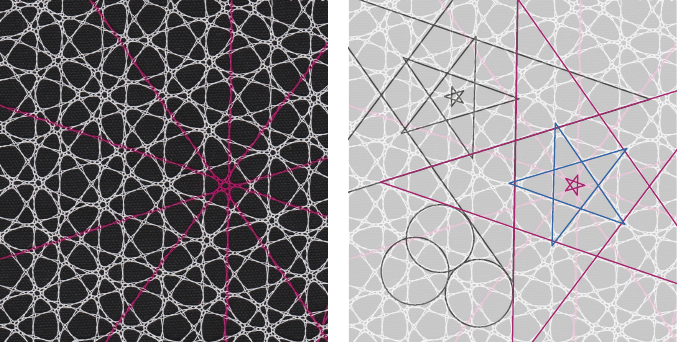}}%
  \end{picture}%
\endgroup%

\vfigend{fig:laceP3bars}{`Ammann's web', Veronika Irvine 2019: Ammann P3 bar pattern worked in DMC~Cordonnet~Special~80 cotton thread.}{Lace Ammann bars of P3}

Using the patterns identified in the previous sections, we created the lace samples shown in Figures~\ref{fig:lace2},~\ref{fig:bigrid},~\ref{fig:laceP3bars}~and~\ref{fig:laceP3dual}.  To the best of our knowledge, these are the first quasiperiodic bobbin lace pieces.

Traditional bobbin lace is monochromatic. Because the regularity in quasiperiodic patterns is less obvious, we use colour to draw attention to some of the repeated elements.  In Figure~\ref{fig:lace2}, the petal shape is emphasized by yellow threads which are wrapped around two continuous white threads.  The wrapping, which is a variation on the traditional grand Venetian cord~\cite{cook}, is performed as the lace is constructed. In Figure~\ref{fig:laceP3bars}, we used hot pink threads to emphasize the five directions of the parallel line sets.  Where the pink threads intersect, they form a pentagram. The pink pentagram is inside a white pentagon which forms the core of a white pentagram, which intersects a larger white pentagon.  These concentric pentagrams continue to the edges of the piece.  Larger copies are a little harder to spot amongst the crisscrossing lines.  In Figure~\ref{fig:laceP3bars}, it is also possible to spot nearly circular shapes formed by copies of the quasi-unit cell.

\vfigbegin
    % width, filename
    \def\svgwidth{0.75\columnwidth}
%% Accompanies image file 'lace_p3_centroid.pdf' (pdf, eps, ps)
\begingroup%
  \makeatletter%
  \providecommand\color[2][]{%
    \errmessage{(Inkscape) Color is used for the text in Inkscape, but the package 'color.sty' is not loaded}%
    \renewcommand\color[2][]{}%
  }%
  \providecommand\transparent[1]{%
    \errmessage{(Inkscape) Transparency is used (non-zero) for the text in Inkscape, but the package 'transparent.sty' is not loaded}%
    \renewcommand\transparent[1]{}%
  }%
  \providecommand\rotatebox[2]{#2}%
  \newcommand*\fsize{\dimexpr\f@size pt\relax}%
  \newcommand*\lineheight[1]{\fontsize{\fsize}{#1\fsize}\selectfont}%
  \ifx\svgwidth\undefined%
    \setlength{\unitlength}{520.82837677bp}%
    \ifx\svgscale\undefined%
      \relax%
    \else%
      \setlength{\unitlength}{\unitlength * \real{\svgscale}}%
    \fi%
  \else%
    \setlength{\unitlength}{\svgwidth}%
  \fi%
  \global\let\svgwidth\undefined%
  \global\let\svgscale\undefined%
  \makeatother%
  \begin{picture}(1,0.47886838)%
    \lineheight{1}%
    \setlength\tabcolsep{0pt}%
    \put(0,0){\includegraphics[width=\unitlength,page=1]{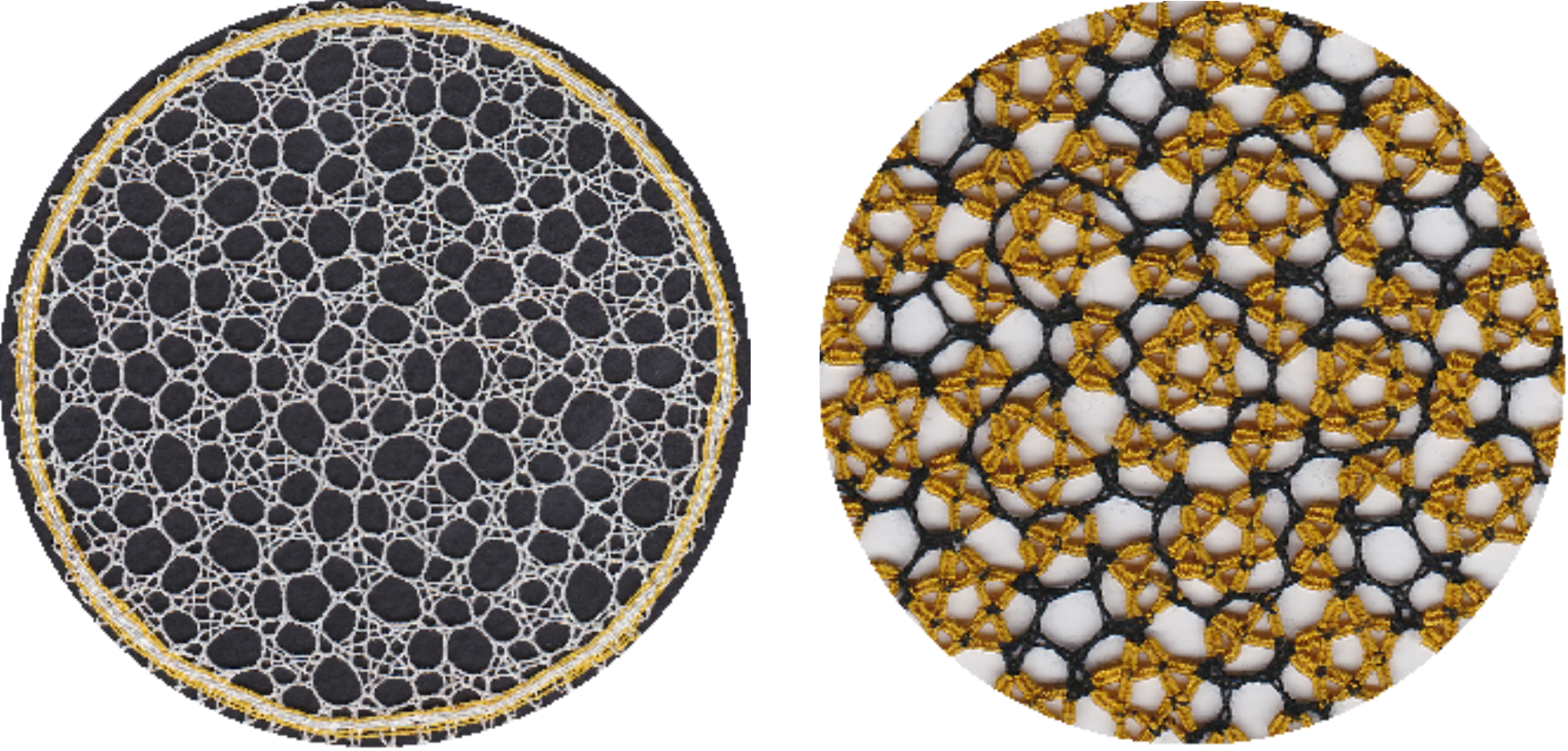}}%
  \end{picture}%
\endgroup%

\vfigend{fig:laceP3dual}{Two variations worked from the centroid-dual of P3 pattern.}{Lace P3}

As a practical observation, we note that the faces in the Ammann bar pattern range significantly in size.  This makes it challenging to choose the scale of the pattern relative to the thickness of the thread being used. Edge lengths in the centroid-dual of P3 are more regular and therefore much easier to work with.

In this paper we have laid the groundwork for exploring non-periodic patterns as bobbin lace grounds by presenting a generalized model.
We have proven the existence of simple quasiperiodic lace patterns based on Sturmian words as well as two richer self-similar patterns based on the P3 tilings and its Ammann bar decoration.  We have conjectured that the larger families to which these patterns belong also satisfy the conditions of bobbin lace.

\section*{Acknowledgements}
We would like to thank Robert Lang for making the recommendation to explore Ammann bars and Latham Boyle for his help with understanding generalized Ammann patterns.  This research was supported by NSERC.

\end{document}